\documentclass[11pt]{article}

\usepackage{amsmath,amssymb,amsbsy,amsfonts,amsthm,latexsym,
            amsopn,amstext,amsxtra,euscript,amscd,color}

\PassOptionsToPackage{hyphens}{url}\usepackage{hyperref}
    
\usepackage[capbesideposition=outside,capbesidesep=quad]{floatrow}

\restylefloat{table}
            
\usepackage{multirow,caption}

\newfont{\teneufm}{eufm10}
\newfont{\seveneufm}{eufm7}
\newfont{\fiveeufm}{eufm5}

\usepackage{systeme}

\usepackage{xcolor}
\usepackage{tikz}
\usetikzlibrary{shapes}
\usetikzlibrary{decorations.markings}

\usepackage{pstricks}
\usetikzlibrary{shapes,arrows}
 \usepackage{pstricks-add}
\usepackage{epsfig}
 \usepackage[space]{grffile} 
 \usepackage{etoolbox} 
 \makeatletter 
 \patchcmd\Gread@eps{\@inputcheck#1 }{\@inputcheck"#1"\relax}{}{}
 \makeatother

\newtheorem{remark}{Remark}

\newtheorem{thm}{Theorem}
\newtheorem{lem}[thm]{Lemma}

\newtheorem{prop}[thm]{Proposition}

\newtheorem{defn}[thm]{Definition}

\newcommand{\Tr}{{\rm Tr}}
\newcommand{\Trn}{{\rm Tr}_n}
\newcommand{\Trm}{{\rm Tr}_m}

\newcommand{\cB}{\mathscr{B}}

\def\+{\oplus}

\def\cB{{\mathcal B}}

\def\cU{{\mathcal U}}

\def\cW{{\mathcal W}}

\def\F{{\mathbb F}}

\def\Fn{{\mathbb{F}_{p^n}}}

\def\00{{\bf 0}}
\def\11{{\bf 1}}
\def\+{\oplus}

\def\\{\cr}
\def\({\left(}
\def\){\right)}

\newcommand{\BBZ}{\mathbb{Z}}
\newcommand{\BBR}{\mathbb{R}}

\newcommand{\bwht}[2]{\mathcal{W}_{#1}(#2)}
\newcommand{\vwht}[3]{\mathcal{W}_{#1}(#2,#3)}

\newcommand{\cardinality}[1]{\# #1}

\providecommand{\newoperator}[3]{%
  \newcommand*{#1}{\mathop{#2}#3}}

\newoperator{\FD}{\mathrm{FD}}{\nolimits}

\begin{document}
\title{\bf Investigations on $c$-Boomerang Uniformity and Perfect Nonlinearity}
\author{Pantelimon~St\u anic\u a  \\ 
Applied Mathematics Department, \\
Naval Postgraduate School, Monterey, USA. \\
E-mail: pstanica@nps.edu}

\maketitle

\begin{abstract}
 We defined in~\cite{EFRST20} a new multiplicative $c$-differential, and the corresponding $c$-differential uniformity and we characterized the known perfect nonlinear functions with respect to this new concept, as well as the inverse in any characteristic. The work was continued in~\cite{RS20}, investigating the $c$-differential uniformity for some further APN functions. Here, we extend the concept to the boomerang uniformity, introduced at Eurocrypt '18 by Cid et al.~\cite{Cid18}, to evaluate S-boxes of block ciphers, and investigate it in the context of perfect nonlinearity and related functions.
\end{abstract}
{\bf Keywords:} 
Boolean, 
$p$-ary functions, 
$c$-differentials,  
differential uniformity,
boomerang uniformity,
perfect and almost perfect $c$-nonlinearity
\newline
{\bf MSC 2000}: 06E30, 11T06, 94A60, 94C10.


\section{Introduction}

In this paper we extend the notion of boomerang uniformity using a previously defined~\cite{EFRST20} multiplier differential (in any characteristic). We characterize some of the known perfect nonlinear functions and the inverse function through this new concept.  We also characterize this concept via the Walsh transforms as Lie et al.~\cite{Li19} did for the classical boomerang uniformity.

The objects of this study are Boolean and $p$-ary functions (where $p$ is an odd prime) and some of their differential properties.  We will introduce here only some needed notation, and the reader can consult~\cite{Bud14,CH1,CH2,CS17,MesnagerBook,Tok15} for more on Boolean and $p$-ary functions.

For a positive integer $n$ and $p$ a prime number, we let $\F_{p^n}$ be the  finite field with $p^n$ elements, and $\F_{p^n}^*=\F_{p^n}\setminus\{0\}$ be the multiplicative group (for $a\neq 0$, we often write $\frac{1}{a}$ to mean the inverse of $a$ in the multiplicative group). We let $\F_p^n$ be the $n$-dimensional vector space over $\F_p$.  We use $\cardinality{S}$ to denote the cardinality of a set $S$ and $\bar z$, for the complex conjugate.
We call a function from $\F_{p^n}$ (or $\F_p^n$) to $\F_p$  a {\em $p$-ary  function} on $n$ variables. For positive integers $n$ and $m$, any map $F:\F_{p^n}\to\F_{p^m}$ (or, $\F_p^n\to\F_p^m$)  is called a {\em vectorial $p$-ary  function}, or {\em $(n,m)$-function}. When $m=n$, $F$ can be uniquely represented as a univariate polynomial over $\F_{p^n}$ (using some identification, via a basis, of the finite field with the vector space) of the form
$
F(x)=\sum_{i=0}^{p^n-1} a_i x^i,\ a_i\in\F_{p^n},
$
whose {\em algebraic degree}   is then the largest Hamming weight of the exponents $i$ with $a_i\neq 0$. 
For $f:\F_{p^n}\to \F_p$ we define the {\it Walsh-Hadamard transform} to be the integer-valued function
$\displaystyle
\bwht{f}{u}  = \sum_{x\in \F_{p^n}}\zeta_p^{f(x)-\Trn(u x)}, \ u \in \mathbb{F}_{p^n},
$
 where $\zeta_p= e^{\frac{2\pi i}{p}}$ and $\Trn:\F_{p^n}\to \F_p$ is the absolute trace function, given by $\displaystyle \Trn(x)=\sum_{i=0}^{n-1} x^{p^i}$ (we will denote it by $\Tr$, if the dimension is clear from the context). The Walsh transform $\vwht{F}{a}{b}$ of an $(n,m)$-function $F$ at $a\in \F_{p^n}, b\in \F_{p^m}$ is the Walsh-Hadamard transform of its component function ${\rm Tr}_m(bF(x))$ at $a$, that is,
\[
  \vwht{F}{a}{b}=\sum_{x\in\F_{p^n}} \zeta_p^{\Trm(bF(x))-\Trn(ax)}.
\]
(If one wishes to work with vector spaces, then one can replace the $\Tr$ by any scalar product on that environment.)
 
Given a $p$-ary  function $f$, the derivative of $f$ with respect to~$a \in \F_{p^n}$ is the $p$-ary  function
$
 D_{a}f(x) =  f(x + a)- f(x), \mbox{ for  all }  x \in \F_{p^n},
$
which can be naturally extended to vectorial $p$-ary functions.

For an $(n,n)$-function $F$, and $a,b\in\F_{p^n}$, we let $\Delta_F(a,b)=\cardinality{\{x\in\F_{p^n} : F(x+a)-F(x)=b\}}$. We call the quantity
$\delta_F=\max\{\Delta_F(a,b)\,:\, a,b\in \F_{p^n}, a\neq 0 \}$ the {\em differential uniformity} of $F$. If $\delta_F= \delta$, then we say that $F$ is differentially $\delta$-uniform. If $\delta=1$, then $F$ is called a {\em perfect nonlinear} ({\em PN}) function, or {\em planar} function. If $\delta=2$, then $F$ is called an {\em almost perfect nonlinear} ({\em APN}) function. It is well known that PN functions do not exist if $p=2$.

The paper is organized as follows. Section~\ref{sec2} contains some background on differential and boomerang uniformity and our extension to $c$-boomerang uniformity. Section~\ref{sec3} contains some characterizations of this concept and connections with prior $c$-differential uniformity, and Section~\ref{sec4} gives its description via Walsh transforms. Section~\ref{sec5} puts together two lemmas needed in the remaining of the paper, and Section~\ref{sec6} deals with perfect nonlinearity and $c$-boomerang uniformity. Section~\ref{sec7} characterizes this uniformity for the inverse function in any characteristic and Section~\ref{sec8} concludes the paper. An appendix with some computational data follows the references.

\section{Differential and boomerang uniformity}
\label{sec2}

Wagner~\cite{Wag99} introduced the boomerang attack against block ciphers using $S$-boxes. The concept was picked up and used in~\cite{BK99,KKS00} on some real ciphers (see also~\cite{BDK02,Kim12}).
Ten years later, however, a theoretical new cryptanalysis tool based upon this attack was introduced at EUROCRYPT 2018 by Cid et al.~\cite{Cid18} namely, the Boomerang Connectivity Table (BCT) and Boomerang Uniformity. In~\cite{Cid18}, the authors analyzed some of the properties of BCT, like its relationship  with the Differential Distribution Table (DDT). They proved that perfect nonlinear (APN) S-boxes (that is, with 2-uniform DDT) always have
2-uniform BCT and for any choice of the parameters, and also showed that the  BCT uniformity is greater than or equal to the DDT uniformity (we will display below precisely the relationship between these).  

The initial concept of boomerang uniformity was defined for permutations (of course, proper $S$-boxes)  in the following way.
\begin{defn}
Let $F$ be a permutation on $\F_{2^n}$ and $(a,b)\in\F_{2^n}\times \F_{2^n}$. We define the entries of the {\em Boomerang Connectivity Table} \textup{(}{\em BCT}\textup{)} by
\[
\cB_F(a,b)=\cardinality \{x\in\F_{2^n}|F^{-1} (F(x)+b)+F^{-1}(F(x+a)+b)=a  \},
\]
where $F^{-1}$ is the compositional inverse of $F$. The {\em boomerang uniformity} of $F$ is defined as
\[
\beta_F=\max_{a,b\in\Fn*} \cB_F(a,b).
\]
We also say that $F$ is a $\beta_F$-uniform BCT function.
\end{defn}

Recently, BCT and the boomerang uniformity were further studied by Boura and Canteaut~\cite{BC18}. They  showed that the boomerang uniformity is only an affine equivalence invariant but not necessarily, extended affine nor CCZ-equivalence invariant.
 Further, they also obtained the boomerang uniformity of the inverse function $F(x)=x^{2^n-2}$ over $\F_{2^n}$, for $n$ even, namely, $\beta_F=4,6$, when $n\equiv 2\pmod 4$, respectively,  $n\equiv 0\pmod 4$. Also, for the Gold function $G(x)=x^{2^t+1}$  over $\F_{2^n}$, where $n\equiv 2\pmod 4$, $t$ even with $\gcd(t,n)=2$, then $\delta_G=\beta_G=4$.

Mesnager et al.~\cite{Mes19} continued the work and showed that the differential uniformity for quadratic functions on $\F_{q^n}$ ($q$ is a $2$-power) is always $\geq q$, and if it happens to be equal to $q$ for a permutations $F$ then its  boomerang uniformity must be $q$, as well. In fact~\cite{Li19}, in general, for quadratic permutations $F$, then we have, 
\[
\delta_F\leq \beta_F\leq \delta_F(\delta_F-1).
\]
It is also easy to show that for monomials $F(x)=x^d$, then $\displaystyle\beta_F=\max_{b\neq 0} \beta_F(1,b)$. For the Bracken-Tan-Tan function (which is an extension of Budaghyan-Carlet function~\cite{BC08}), $F(x)=\alpha x^{2^s+1}+\alpha^{2^k} x^{2^{-k}+2^{k+s}}$ (under some conditions on the parameters), then $\beta_F=4$ (see~\cite{Mes19}). More work has been done recently on BCT, and we mention here~\cite{BPT19,CV19,LiHu20,TX20}.

Surely, $\Delta_F(a,b)=0,2^n$ and $\cB_F(a,b)=2^n$  whenever $ab=0$.
It is well-known from prior work that $\delta_F=\delta_{F^{-1}}$ and $\beta_F=\beta_{F^{-1}}$. In general, for permutations, $\beta_F\geq \delta_F$ and they are equal for APN permutations. A natural question is what is the algebraic   difference between these two concepts and Boura and Canteaut answered that question in~\cite{BC18}, showing that
\[
\cB_F(a,b)=\Delta_F(a,b)+\sum_{\gamma\in\Fn^*, \gamma\neq b} \cardinality\left(\cU_{\gamma,a}^{F^{-1}}\cap \left(b+\cU_{\gamma,a}^{F^{-1}}\right)\right),
\]
where $\cU_{\gamma,a}^{F^{-1}}=\{x\in\Fn | D_{\gamma} F^{-1} (x)=a\}$.
This was reformulated by Mesnager et al.~\cite{Mes19} in the following way:
\[
\cB_F(a,b)=\sum_{\gamma\in\Fn^*} \cardinality\left(\cU_{\gamma,a}^F\cap \left(b+\cU_{\gamma,a}^F \right)\right),
\]
where $\cU_{\gamma,a}^F=\{x\in\Fn|D_{\gamma} F  (x)=a\}$,  and 
further, by Li et al.~\cite{Li19}, as 
\[
\cB_F(a,b)=\cardinality \left\{ (x,y)\in\Fn\times \Fn\,\Large\big|\, \substack{F(x)+F(y)=b \\  F(x+a)+F(y+a)=b } \right\}
\]
and easily observed in~\cite{Mes19} that this transforms into (labeling $y=x+\gamma$),
\begin{equation}
\begin{split}
\label{eq:boom-diff}
\cB_F(a,b)&=\cardinality \left\{ (x,\gamma)\in\Fn\times \Fn\,\Large\big|\, \substack{F(x+\gamma)+F(x)=b \\ F(x+\gamma+a)+F(x+a)=b }   \right\}\\
& =\sum_{\gamma\in\Fn} \cardinality  \left\{ x \in\Fn \,\big|\, D_\gamma F(x)=b \text{ and }  D_\gamma F(x+a)=b   \right\}.
\end{split}
\end{equation}
Avoiding the inverse of $F$, allows these last expressions to define the boomerang uniformity for  functions that are not necessarily permutations.

Before, we continue with our approach, let us recall the concept we (along with others) introduced in~\cite{EFRST20} (we will define it in general on $\F_{p^n}$, $p$ prime, not only for $p=2$).

Inspired by a practical differential attack developed  in~\cite{BCJW02} (though, via a different differential), we extended the definition of derivative and differential uniformity in ~\cite{EFRST20}, in the following way.
For a $p$-ary $(n,m)$-function   $F:\F_{p^n}\to \F_{p^m}$, and $c\in\F_{p^m}$, the ({\em multiplicative}) {\em $c$-derivative} of $F$ with respect to~$a \in \F_{p^n}$ is the  function
\[
 _cD_{a}F(x) =  F(x + a)- cF(x), \mbox{ for  all }  x \in \F_{p^n}.
\]
(Observe that, if   $c=1$, then we obtain the usual derivative, and, if $c=0$ or $a=0$, then we obtain a shift of the function, in the input/output.)

For an $(n,n)$-function $F$, and $a,b\in\F_{p^n}$, we let the entries of the $c$-Difference Distribution Table ($c$-DDT) be defined by ${_c\Delta}_F(a,b)=\cardinality{\{x\in\F_{p^n} : F(x+a)-cF(x)=b\}}$. We call the quantity
\[
\delta_{F,c}=\max\left\{{_c\Delta}_F(a,b)\,|\, a,b\in \F_{p^n}, \text{ and } a\neq 0 \text{ if $c=1$} \right\}\]
the {\em $c$-differential uniformity} of~$F$ (observe that we slightly change here the way we denoted the  $c$-differential uniformity in~\cite{EFRST20}). If $\delta_{F,c}=\delta$, then we say that $F$ is differentially $(c,\delta)$-uniform (or that $F$ has $c$-uniformity $\delta$, or for short, {\em $F$ is $\delta$-uniform $c$-DDT}). If $\delta=1$, then $F$ is called a {\em perfect $c$-nonlinear} ({\em PcN}) function (certainly, for $c=1$, they only exist for odd characteristic $p$; however, as proven in~\cite{EFRST20}, there exist PcN functions for $p=2$, for all  $c\neq1$). If $\delta=2$, then $F$ is called an {\em almost perfect $c$-nonlinear} ({\em APcN}) function. 
When we need to specify the constant $c$ for which the function is PcN or APcN, then we may use the notation $c$-PN, or $c$-APN.
It is easy to see that if $F$ is an $(n,n)$-function, that is, $F:\F_{p^n}\to\F_{p^n}$, then $F$ is PcN if and only if $_cD_a F$ is a permutation polynomial.

\begin{figure}[ht]
\centering
  \includegraphics[width=.6\linewidth]{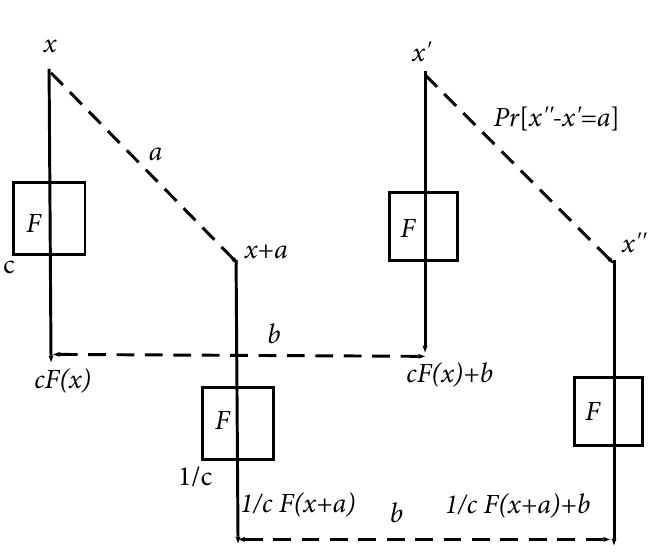}
  \caption{View on the $c$-boomerang uniformity}
  \label{fig:cboom}
\end{figure}

In light of the formulations~\eqref{eq:boom-diff} via differentials, by using our $c$-differential concept of~\cite{EFRST20}, we extend the notion of boomerang uniformity to $c$-boomerang uniformity for all characteristics, in the following way (see Figure~\ref{fig:cboom}). 
\begin{defn}
For an $(n,n)$-function $F$, $c\neq 0$, and $(a,b)\in\Fn\times \Fn$,  we define the {\em $c$-Boomerang Connectivity Table}  \textup{(}$c$-BCT\textup{)} entry at $(a,b)$ to be
{\small
\begin{equation}
\label{eq:originalBCT}
  _c\cB_F(a,b)=\cardinality \left\{ x\in\Fn\,\Big|\,  F^{-1}(c^{-1} F(x+a)+b) -F^{-1}(cF(x)+b)=a \right\}.
\end{equation}
}
Further, the {\em $c$-boomerang uniformity} of $F$ is defined by
\[
\beta_{F,c}=\max_{a,b\in\Fn*} {_c}\cB_F(a,b).
\]
If $\beta_{F,c}=\beta$, we also say that $F$ is a $\beta$-uniform $c$-BCT function.
\end{defn}

\section{Characterizations of $c$-boomerang uniformity}
\label{sec3}

As in the classical case, we find an alternative formulation that avoids inverses, allowing the definition to be extended to all $(n,n)$-function, not only permutations. Though the proof is not complicated, we label as a theorem, since it is the way to extend the definition to non-permutations.
\begin{thm}
\label{thm:new-c-boom}
For an $(n,n)$-permutation function $F$, $c\neq 0$,  the entries of the $c$-Boomerang Connectivity Table at $(a,b)\in\F_{p^n}\times\F_{p^n}$  are given by
\begin{align*}
_c\cB_F(a,b)&=\cardinality \left\{ (x,\gamma)\in\Fn\times \Fn\,\Big|\, \Large\substack{F(x+\gamma)-cF(x) =b \\  F(x+\gamma+a)- c^{-1} F(x+a)=b }  \right\}\\
& =\sum_{\gamma\in\Fn} \cardinality  \left\{ x \in\Fn \,\big|\, _cD_\gamma F(x)=b \text{ and }  _{c^{-1}}D_\gamma F(x+a)=b   \right\}.
\end{align*}
\textup{(}We shall call the system above, the $c$-boomerang system, for easy referral.\textup{)}
\end{thm}
\begin{proof}
Let $y=F^{-1}(cF(x)+b)$ in \eqref{eq:originalBCT}. If that is the case, then $F(y)-cF(x)=b$. Further, from~\eqref{eq:originalBCT}, we obtain,
\[
F^{-1}(c^{-1} F(x+a)+b)=y+a,
\]
and so,
\[
c^{-1}F(x+a)+b =F(y+a)\ \Longleftrightarrow \ F(y+a)-c^{-1}F(x+a)=b.
\]
The reciprocal  is immediate by backtracking the argument:
the second equation in the system implies $F(y+a)=c^{-1}F(x+a)+b$ and applying $F^{-1}$, we get
$F^{-1}(c^{-1}F(x+a)+b )=y+a=F^{-1}(F(y))+a=F^{-1}(cF(x)+b)+a$, using the first equation.
 Finally, taking $y=x+\gamma$,  the theorem follows.
\end{proof}

\begin{remark}
We could have defined the BCT entries of $F$ at $(a,b)$ to be 
\[
{_c}\cB_F(a,b)=\cardinality \left\{ (x,\gamma)\in\Fn\times \Fn\,\Big|\, \Large\substack{F(x+\gamma)-cF(x) =b \\  cF(x+\gamma+a)- F(x+a)=cb }  \right\},
\]
 thus allowing $c=0$. However, in the case of $c=0$, we would obtain
${_0}\cB_F(a,b)=\cardinality \left\{ (x,\gamma)\in\Fn\times \Fn\,\big|\, F(x+\gamma) =b \text{ and } F(x+a)=0  \right\}$;  when $F$ is a permutation, we  get $x=F^{-1}(0)-a$ and $\gamma=F^{-1}(b)-F^{-1}(0)+a$, rendering $\beta_{F,0}=1$, regardless of the function. For this reason, we decided to remove $c=0$ out of allowable multipliers.
\end{remark}

\begin{remark}
For whatever reason, one can define the $c$-DDT and $c$-BCT, and the respective uniformities, even for arbitrary $(n,m)$-functions. That is the reason why in some of our results, when it makes no difference for the proofs (except tracking carefully the parameters, of course), we take arbitrary $(n,m)$-functions.
\end{remark}

While we already knew that,  in general, the $c$-differential uniformity of $F,F^{-1}$ are not the same, we 
 show below a connection (that we did not observe in~\cite{EFRST20}) between some of the entries in the $c$-Differential Distribution Table of a permutation monomial $F$ and the one of the inverse of $F$.
\begin{prop}
Let $F(x)=x^d$ be a permutation  $($hence, $\gcd(d,p^n-1)=1$$)$ monomial function on~$\Fn$ and $c\neq 0$. Then, 
\[
{_c}\Delta_{F^{-1}} (a,b)={_{c^{-d} }}\Delta_F(bc^{-1},-a\,c^{-d} ).
\]
\end{prop}
\begin{proof}
Observe that, in general (assuming $c\neq 0$),
\allowdisplaybreaks
\begin{align*}
& F^{-1}(x+a)-c F^{-1}(x)=b  \Longleftrightarrow\\
&F( F^{-1}(x+a))= F(c F^{-1}(x)+b) \Longleftrightarrow\\
& F(c F^{-1}(x)+b) -x=a \Longleftrightarrow\\
& F(F^{-1}(x))- F(c F^{-1}(x)+b) =a.
\end{align*}
If $F$ is the monomial $x^d$, then  the last equation above becomes 
\begin{align*}
& F(F^{-1}(x)) - c^d F( F^{-1}(x)+bc^{-1}) =a  \Longleftrightarrow\\
&  F( F^{-1}(x)+bc^{-1})- c^{-d} F(F^{-1}(x))=-a c^{-d},
\end{align*}
and the proposition will follow easily.
\end{proof}

We next show a connection between $c$-differential and $c$-boomerang uniformities.
 \begin{thm}
Let $F$ be a  permutation function  on $\Fn$, $(a,b)\in\Fn\times \Fn$ and $c\neq 0,1$. Then
${_c}\cB_F(0,b)={{_c\Delta}}_F(0,b)$. If $p=2$, then
\begin{align*}
 {_c}\cB_F(a,b)&={_c}\Delta_F(a,b)+\sum_{a\neq \gamma\in\Fn} \cardinality  \left\{ x \in\Fn \,\big|\, \substack{_cD_\gamma F(x)=b \\ {_{c^{-1}}}D_\gamma F(x+a)=b }  \right\},
 \end{align*}
 if either $c=-1$ or $b=0$.
Consequently, $\beta_{F,-1}\geq \delta_{F,-1}$, for all permutations~$F$ \textup{(}it is known~\textup{\cite{BC18}} that $\beta_{F,1}\geq \delta_{F,1}$\textup{)}.
\end{thm}
\begin{proof}
If $a=0$ and $c\neq 0,\pm1$, then 
\begin{align*}
_c\cB_F(0,b)&=\cardinality \left\{ (x,\gamma)\in\Fn\times \Fn\,\Big|\, \Large\substack{F(x+\gamma)-cF(x) =b \\ F(x+\gamma)-c^{-1}F(x)=b } \right\},
\end{align*}
and so, $F(x)=0$ and $F(x+\gamma)=b$, which has only one solution $(x,\gamma)=(F^{-1}(0),F^{-1}(b)-F^{-1}(0))$.  
Further, ${{_c\Delta}}_F(0,b)=\cardinality\{x\in\Fn\,|\, {_c}D_\gamma F(x)=F(x)-cF(x)=b\}$, which also has one solution, and so, ${_c}\cB_F(0,b)={{_c\Delta}}_F(0,b)$. When $a=0$ and $c=-1$, the values of the two uniformities are the same, since 
\begin{align*}
{_{-1}}\cB_F(0,b)&=\cardinality \left\{ (x,\gamma)\in\Fn\times \Fn\,\Big|\, \Large\substack{F(x+\gamma)+F(x) =b \\ F(x+\gamma)+F(x)=b } \right\}\\
&= \cardinality \left\{ (x,\gamma)\in\Fn\times \Fn\,\Big|\,  F(x+\gamma)+F(x) =b \right\}=
{_{-1}}\Delta_F(0,b),
\end{align*}
and the first claim is shown.

Assume now that $a\neq 0$. If $\gamma=0$, then (we use the logical conjunction ``$\wedge$'' to denote ``and'').
\begin{align*}
&  
\left\{ x \in\Fn \,\Big|\,  F(x)-cF(x)=b \wedge F(x+a)-c^{-1}F(x+a)=b \right\}\\
=& \left\{ x \in\Fn \,\big|\,  F(x)=\frac{b}{1-c}  \wedge F(x+a) =\frac{b}{1-c^{-1}}  \right\}.
 \end{align*}
Thus, for $c$ fixed, we get $x=F^{-1} \left (\frac{b}{1-c}  \right)$ and so, $a=F^{-1} \left (\frac{b}{1-c^{-1}}  \right)-F^{-1} \left (\frac{b}{1-c} \right)$, since  $F$ is a permutation. Thus, when $\gamma= 0$,
\[
\cardinality  \left\{ x \in\Fn \,\Big|\, \Large\substack{F(x)-cF(x)=b \\ F(x+a)-c^{-1}F(x+a)=b }\right\}=\begin{cases}
1& \text{ if } a=F^{-1} \left (\frac{b}{1-c^{-1}}  \right)-F^{-1} \left (\frac{b}{1-c} \right)\\
0 & \text{ otherwise}.
\end{cases}
\]

 If $\gamma=a$ and $p=2$, then 
 \begin{align*}
&  \left\{ x \in\Fn \,\Large\big|\, \substack{F(x+a)-cF(x)=b \\ F(x)-c^{-1} F(x+a)=b }\right\}=  \left\{ x \in\Fn \,\Large\big|\, \substack{F(x+a)-cF(x)=b \\  F(x+a)-cF(x)=-bc }\right\},
 \end{align*}
which is $\emptyset$, unless $b=0$, or $c=-1$, in which case we have
 \begin{align*}
 &\cardinality  \left\{ x \in\Fn \,\large\big|\,  F(x+a)-cF(x)=b  \right\}  ={_c}\Delta_F (a,b),
 \end{align*}
 Thus, 
 \begin{align*}
 {_c}\cB_F(a,b)&={_c}\Delta_F(a,b)+\sum_{a\neq \gamma\in\Fn} \cardinality  \left\{ x \in\Fn \,\big|\, \substack{_cD_\gamma F(x)=b \\  _{c^{-1}}D_\gamma F(x+a)=b }  \right\},
 \end{align*}
 if either $c=-1$ or $b=0$.
\end{proof}

\section{Characterizing $c$-boomerang uniformity  via the Walsh transform}
\label{sec4}

Using a method of Carlet~\cite{Car18} and Chabaud and Vaudenay~\cite{CV95} connecting the differential uniformity of an $(n,m)$-function to its Walsh coefficients, we characterized the $c$-differential uniformity in~\cite{EFRST20} and Li et al.~\cite{Li19} characterized the boomerang uniformity. We can use a similar method to do the same for the $c$-boomerang uniformity, in any characteristic. 

\begin{thm}
\label{thm:charact2}
Let $c\in\F_{p^m}^*$ and $n,m,\delta$ be fixed  positive integers. Let $F$  be an $(n,m)$-function,   $F:\F_{p^n}\to\F_{p^m}$ and $\phi_\beta(x)=\sum_{j\geq 0} A_j x^j$ be a polynomial over $\BBR$ such that $\phi_\beta(x)=0$ for $x\in\BBZ, 1\leq x\leq \beta$, and $\phi_\beta(x)>0$, for $x\in\BBZ, x> \beta$.  We then have
  \begin{align*}
p^{m+n} A_0+&\sum_{j\geq 1}  p^{-(2j-1)(m+n)} A_j \sum_{\substack{w_1,\ldots,w_j\in \F_{p^n}\\
z_1,\ldots,z_j\in \F_{p^n}\\
u_1,\ldots,u_j\in \F_{p^m}\\
v_1,\ldots,v_j\in \F_{p^m}\\
\sum_{i=1}^j (w_i+z_i)=0\\
\sum_{i=1}^j (u_i+v_i)=0
}}\\
&
 \qquad \prod_{i=1}^j \left( \cW_F(u_i,z_i)\overline{\cW_F}(cu_i,-w_i) \cW_F(v_i,-z_i) \overline{\cW_F}(c^{-1} v_i, w_i) \right)\geq 0,
\end{align*}
with equality if and only if $F$ is $\beta$-uniform $c$-BCT.
\end{thm}
\begin{proof}
We follow mostly~\cite{Car18,EFRST20,Li19}, pointing out the differences, where applicable.
Let $a\in\F_{p^n},d\in\F_{p^m}$ be arbitrary elements.  We let $n_F(a,b,c)=\cardinality \{ (x,\gamma)\in\F_{p^n}\times\F_{p^n}\,|\, _cD_\gamma F(x)= {_c}D_\gamma F(b), {_{c^{-1}}}D_\gamma F(x+a)= {_c}D_\gamma F(b)\} =\cardinality \{ (x,y)\in\F_{p^n}\times\F_{p^n}\,|\, F(y)-cF(x)= b, F(y+a)-c^{-1}F(x+a)= b\}  \geq 0$. From the definition of $\phi_\beta$, for $c$ fixed, and for all $a,b\in\F_{p^n}$, then,
\[
\displaystyle \sum_{j\geq 0} A_j \left(n_F(a,b,c) \right)^j\geq 0,
\]
with equality if and only if $_c\cB_F(a,b)=\beta$. Consequently, running with all $a,b\in\F_{p^n}$, any $(n,m)$-function $F$ satisfies
\[
\sum_{j\geq 0} A_j \sum_{a,b\in\F_{p^n}} \left(n_F(a,b,c) \right)^j\geq 0,
\]
with equality if and only of $\beta_{F,c}=\beta$.

Using the well-known, 
\begin{align}
\label{eq:trace}
  \sum_{v\in\F_{p^m}} \zeta_p^{\Trm(v\alpha)}=
\begin{cases}  0 & \text{ if } \alpha\neq 0\\
p^m & \text{ if } \alpha = 0,
\end{cases}
\end{align} 
we see that
\[
\displaystyle n_F(a,b,c)=p^{-2m} \sum_{x,y\in F_{p^n}, u,v\in F_{p^m}} \zeta_p^{\Trm(u(F(y)-cF(x)-b))+\Trm(v ( F(y+a)-c^{-1}F(x+a)- b))},
\]
which, when run for all $j$ tuples of parameters ($j\geq 1$ is fixed) renders
 \begin{align*}
&\sum_{a\in\F_{p^n},b\in\F_{p^m}} \left(n_F(a,b,c) \right)^j \\
&= p^{-2jm}\sum_{\substack{a\in\F_{p^n}\\ b\in\F_{p^m}}}\sum_{\substack{x_1,\ldots,x_j\in\F_{p^n}\\ y_1,\ldots,y_j\in\F_{p^n}\\ u_1,\ldots,u_j\in \F_{p^m}\\ v_1,\ldots,v_j\in \F_{p^m}}} \zeta_p^{\sum_{i=1}^j \left[\Trm(u_i(F(y_i)-cF(x_i)-b))+\Trm(v_i ( F(y_i+a)-c^{-1}F(x_i+a)- b))\right]}.
\end{align*}
By~\eqref{eq:trace},  $\displaystyle \sum_{w_i\in\F_{p^n}} \zeta_p^{\Trn (w_i(x_i'-x_i-a))}=p^n$, if $x_i'=x_i+a$ and $0$, otherwise; similarly, $\displaystyle \sum_{z_i\in\F_{p^n}} \zeta_p^{\Trn (z_i(y_i'-y_i-a))}=p^n$, if $y_i'=y_i+a$ and $0$, otherwise.
Therefore,
\allowdisplaybreaks
\begin{align*}
&\sum_{ a\in\F_{p^n},b\in\F_{p^m}} \left(n_F(a,b,c) \right)^j =p^{-2mj} \ p^{-2nj} \sum_{ a\in\F_{p^n},b\in\F_{p^m}} \sum_{\substack{x_1,\ldots,x_j, x_1',\ldots,x_j'\in\F_{p^n}\\
y_1,\ldots,y_j,y_1',\ldots,y_j'\in\F_{p^n}\\ 
w_1,\ldots,w_j,z_1,\ldots,z_j\in \F_{p^n}\\
u_1,\ldots,u_j,v_1,\ldots,v_j\in \F_{p^m}}} \\
&\zeta_p^{\sum_{i=1}^j \left[\Trm(u_i (F(y_i)-cF(x_i)-b))+\Trm(v_i ( F(y_i')-c^{-1}F(x_i')- b))+\Trn(w_i(x_i'-x_i-a))+\Trn(z_i(y_i'-y_i-a))\right]}\\
&=p^{-2(m+n)j}  \sum_{\substack{w_1,\ldots,w_j,z_1,\ldots,z_j\in \F_{p^n}\\
u_1,\ldots,u_j,v_1,\ldots,v_j\in \F_{p^m}}}
 \prod_{i=1}^j \left( \cW_F(u_i,z_i)\overline{\cW_F}(cu_i,-w_i) \cW_F(v_i,-z_i) \overline{\cW_F}(c^{-1} v_i, w_i) \right)\\
 &\qquad\qquad\qquad\qquad\qquad\qquad\qquad\qquad \sum_{a\in\F_{p^n},b\in\F_{p^m}} \zeta_p^{\Trm \left(-a\sum_{i=1}^j (w_i+z_i) -b\sum_{i=1}^j (u_i+v_i)\right)}\\
&=p^{-(2j-1)(m+n)} \sum_{\substack{w_1,\ldots,w_j\in \F_{p^n}\\
z_1,\ldots,z_j\in \F_{p^n}\\
u_1,\ldots,u_j\in \F_{p^m}\\
v_1,\ldots,v_j\in \F_{p^m}\\
\sum_{i=1}^j (w_i+z_i)=0\\
\sum_{i=1}^j (u_i+v_i)=0
}}
 \prod_{i=1}^j \left( \cW_F(u_i,z_i)\overline{\cW_F}(cu_i,-w_i) \cW_F(v_i,-z_i) \overline{\cW_F}(c^{-1} v_i, w_i) \right).
\end{align*}
Surely, if $j=0$,  $\displaystyle \sum_{\substack{a\in\F_{p^n}\\
b\in\F_{p^m}} }\left(n_F(a,b,c) \right)^j =p^{m+n}$ and
the theorem follows.
\end{proof}

As a particular case, we want to characterize the $1$-uniform $c$-BCT functions.
We can take the polynomial $\phi_1(x)=x-1$, which certainly satisfies the conditions of Theorem~\ref{thm:charact2}. Thus $A_0=-1,A_1=1$ and the relation of Theorem~\ref{thm:charact2}  simplifies to
\begin{align*}
&-p^{n+m}+p^{-(n+m)} \sum_{\substack{w,z\in\F_{p^n}\\
u,v\in\F_{p^m}\\
w=-z,u=-v
} }\cW_F(u,z)\overline{\cW_F}(cu,-w) \cW_F(v,-z) \overline{\cW_F}(c^{-1} v, w)\\
&=\sum_{\substack{z\in\F_{p^n}\\
v\in\F_{p^m}
} }\cW_F(-v,z)\overline{\cW_F}(-cv,z) \cW_F(v,-z) \overline{\cW_F}(c^{-1} v, -z) \geq 0.
\end{align*}
Thus, we obtain the next result.
\begin{prop}
Let  $m,n$ be fixed  positive integers and $c\in\F_{p^m}$, $c\neq 0,1$. Let $F$  be an $(n,m)$-function. Then
\[
\sum_{\substack{z\in\F_{p^n}\\v\in\F_{p^m}
} }\cW_F(-v,z)\overline{\cW_F}(-cv,z) \cW_F(v,-z) \overline{\cW_F}(c^{-1} v, -z) \geq p^{2(n+m)},
 \]
 with equality if and only if $F$ is a $1$-uniform $c$-BCT function.
\end{prop}


\section{Some needed lemmas}
\label{sec5}

In the next few sections, we will investigate some known perfect nonlinear, as well as the  inverse function (in all characteristics) with respect to the $c$-boomerang uniformity.
We will need the following two lemmas. The proof of Lemma~\ref{lem10}$(i)$ can be found in~\cite{BRS67} and   Lemma~\ref{lem10}$(ii)$ is easy and argued in~\cite{EFRST20}. The proof of Lemma~\ref{lem:gcd} is contained in~\cite{EFRST20}.
\begin{lem}
\label{lem10}
Let $n$ be a positive integer. We have:
\begin{enumerate}
 \item[$(i)$] The equation
$x^2 + ax + b = 0$, with $a,b\in\F_{2^n}$, $a\neq 0$,
has two solutions in $\F_{2^n}$ if  $\Tr\left(
\frac{b}{a^2}\right)=0$, and zero solutions otherwise.
\item[$(ii)$]  The equation
$x^2 + ax + b = 0$, with $a,b\in\F_{p^n}$, $p$ odd,
has (two, respectively, one) solutions in $\F_{p^n}$ if and only if the discriminant $a^2-4b$ is a (nonzero, respectively, zero) square in $\F_{p^n}$.
\end{enumerate}
\end{lem}

\begin{lem}
\label{lem:gcd}
Let $p,k,n$ be integers greater than or equal to $1$ (we take $k\leq n$, though the result can be shown in general). Then
\begin{align*}
&  \gcd(2^{k}+1,2^n-1)=\frac{2^{\gcd(2k,n)}-1}{2^{\gcd(k,n)}-1},  \text{ and if  $p>2$, then}, \\
& \gcd(p^{k}+1,p^n-1)=2,   \text{ if $\frac{n}{\gcd(n,k)}$  is odd},\\
& \gcd(p^{k}+1,p^n-1)=p^{\gcd(k,n)}+1,\text{ if $\frac{n}{\gcd(n,k)}$ is even}.\end{align*}
Consequently, if either $n$ is odd, or $n\equiv 2\pmod 4$ and $k$ is even,   then $\gcd(2^k+1,2^n-1)=1$ and $\gcd(p^k+1,p^n-1)=2$, if $p>2$.
\end{lem}
 We will be using  throughout Hilbert's Theorem 90 (see~\cite{Bo90}), which states that if $\mathbb{F}\hookrightarrow \mathbb{K}$  is a cyclic Galois extension and $\sigma$ is a generator of the Galois group ${\rm Gal}(\mathbb{K}/\mathbb{F})$, then for $x\in \mathbb{K}$, the relative trace $\Tr_{\mathbb{K}/\mathbb{F}}(x)=0$ if and only if $x=\sigma(y)-y$, for some $y\in\mathbb{K}$.

\section{Perfect nonlinearity and $c$-boomerang uniformity}
\label{sec6}

The following are some of the known (see, for instance,~\cite{CS97,DY06})   classes of PN  functions ($p$ must be odd).
\begin{thm}
\label{thm-PN}
The following functions $:\F_{p^n}\to\F_{p^n}$ are perfect nonlinear:
\begin{itemize}
\item[$(1)$] $F(x)=x^2$ on $\F_{p^n}$.
\item[$(2)$]  $F(x)=x^{p^k+1}$ on $\F_{p^n}$  is PN if and only if $\frac{n}{\gcd(k,n)}$ is odd.
\item[$(3)$] $F(x)=x^{(3^k+1)/2}$  is PN over $\F_{3^n}$  if and only if $\gcd(k,n)=1$ and $n$ is odd.
\item[$(4)$] $F(x)=x^{10}\pm x^6 - x^2$ is PN over $\F_{3^n}$ if and only if $n=2$ or $n$ is odd. In general, for $u\in \F_{3^n}$, $F(x)=x^{10}-u x^6 - u^2 x^2$ is PN  over $\F_{3^n}$ if $n$ is odd.
\end{itemize}
\end{thm}
It is known that the boomerang uniformity equals the differential uniformity for perfect nonlinear functions. It is, of course, a natural question to ask what is the connection between these in the ``$c$-context''. The reader is pointed to~\cite{EFRST20} where we discussed the $c$-differential uniformity of the functions in Theorem~\ref{thm-PN}.
We will now concentrate on the $c$-boomerang uniformity of all of these classes (for the last function we will just provide some computational data in the appendix). For $c\in\F_{3^n}$, and $b$ fixed, we  let $\mu_c$ denote the cardinality
\begin{equation}
\label{eq:Cheb3}
\cardinality\left\{(x,\gamma)\,\Bigg|\,  \substack{2T_{\frac{3^k+1}{2}}(x+\gamma-1)-(c-c^{-1})T_{\frac{3^k+1}{2}}(x-1)+(c+c^{-1})T_{\frac{3^k-1}{2}}(x-1)=2b\\
 2T_{\frac{3^k-1}{2}}(x+\gamma-1)+(c+c^{-1})T_{\frac{3^k+1}{2}}(x-1)+(c^{-1}-c)T_{\frac{3^k-1}{2}}(x-1)=0}\right\},
\end{equation}
where $T_\ell(w)=u^\ell+u^{-\ell}$, under $w=u+u^{-1}$, is the Chebyshev polynomial of the first kind.
We will be using below the trivial identity, $T_{2\ell} (w)+2=T_{\ell}^2(w)$.
\begin{thm}\label{c-diff}
Let $F:\F_{p^n}\to \F_{p^n}$  $($$p$ is an odd prime number$)$ be the monomial $F(x)=x^d$,  and $c\neq  0, 1$ be fixed. The following statements hold:
\begin{enumerate}
\item[$(i)$] If $d=2$, then $\beta_{F,c}\leq 4$.
\item[$(ii)$] If $d=p^k+1$,   $\delta_{d}:=\cardinality\{\gamma\in\F_{p^n}\,|\, z^{p^k}+z+d=0\}$, then 
  \[
  \beta_{F,c}\geq \delta_{1-c^{-1}}\cdot \left(\delta_{1+c}+1 \right).
  \]
 Moreover, when ${n}/{\gcd{(n,k)}}$ is odd, then $\beta_{F,c}\geq 2$, and when ${n}/{\gcd{(n,k)}}$ is even and $c^{-1}=z^{p^k}+z$, for some $z\neq 0$, then the $c$-boomerang uniformity of $f$ is $\beta_{F,c}\geq p^g$, where $g=\gcd(n,k)$.
\item[$(iii)$]   Let $p=3$. If $\displaystyle d=\frac{3^k+1}{2}$ and $ab\neq 0$, then  
${_c}\cB_F(a,b)\geq \mu_c$, where
$\mu_c$ is defined in~\eqref{eq:Cheb3}. In particular, when $c=-1$, then 
\[
{_c}\cB_F(a,b)\geq \cardinality\left\{(x,\gamma)\,\Bigg|\,  \substack{T_{\frac{3^k+1}{2}}(x+\gamma-1)+T_{\frac{3^k-1}{2}}(x-1)=2b\\
 T_{\frac{3^k-1}{2}}(x+\gamma-1)-T_{\frac{3^k+1}{2}}(x-1)=0}\right\}.
\]

\end{enumerate}
\end{thm}
\begin{proof}
Let $d=2$ and consider the  system ${_c}D_\gamma F(x)=b,{_{c^{-1}}}D_\gamma F(x+a)=b$. 
We do not care for the cases when $ab=0$. 
If $ab\neq 0$, dividing by $a^2$ and relabeling, we may assume that $a=1$. Next, subtracting the first from the second equation, we obtain  
\begin{align*}
 (x+\gamma)=\frac{c^{-1}(x+1)^2-cx^2-1}{2},
\end{align*}
which, when replaced back into the first equation, and expanded, renders  
\begin{align*}
&(1 - 2 c^2 + c^4) x^4 + (4 -     4 c^2) x^3 + (6 - 2 c - 2 c^2 - 2 c^3) x^2\\
&\qquad\qquad\qquad + (4 - 4 c) x + (1 - 2 c + c^2 - 4 b c^2 )  =0. 
\end{align*}
 which has at most four roots, so ${_c}\cB_F(a,b)\leq 4$. We will see in the appendix that all values occur for the first few cases we considered. 

We now consider the Gold case,  $d=p^k+1$. The $c$-boomerang system ${_c}D_\gamma F(x)=b,\ {_{c^{-1}}}D_\gamma F(x+a)=b$ (dividing by $a^{p^{k+1}} \neq 0$ and relabeling, we can assume that $a=1$) becomes
\begin{align*}
&(x+\gamma)^{p^k+1}-c x^{p^k+1}=b\\
&(x+\gamma)^{p^k+1}-c^{-1}(x+1)^{p^k+1}+(x+\gamma)^{p^k} + (x+\gamma)+1=b,
\end{align*}
 eliminating $(x+\gamma)^{p^k+1}$ and expanding the remaining powers, renders
\begin{align}
& cx^{p^k+1}-c^{-1} (x^{p^k+1}+x^{p^k}+x+1) +x^{p^k}+\gamma^{p^k}+x+\gamma+1=0, \text{ or},\nonumber\\
&(c-c^{-1})x^{p^k+1}+(1-c^{-1} )x^{p^k}+(1-c^{-1})x+ \left( \gamma^{p^k}+\gamma+(1-c^{-1})\right)=0.\label{eq:x-gamma}
\end{align}

The idea is to vanish the parenthesis containing $\gamma$, that is, $\gamma^{p^k}+\gamma+(1-c^{-1})=0$ and the polynomial in $x$,
namely, $(c-c^{-1})x^{p^k+1}+(1-c^{-1} )x^{p^k}+(1-c^{-1})x=0$, which is equivalent to 
 $x\left((c+1)x^{p^k}+ x^{p^k-1}+1\right)=0$. By relabeling $x\mapsto 1/x$, the second factor can be put into the form
\[
 x^{p^k}+ x+(c+1)=0,
\]
For more accurate count, we let 
$\delta_{d}=\cardinality\{\gamma\,|\, z^{p^k}+z+d=0\}$. 
We easily infer now that, if $ab\neq 0$, ${_c}\cB_F(a,b)\geq \delta_{1-c^{-1}}\cdot \left(\delta_{1+c}+1 \right)$.

Now, to show the second claim of $(ii)$, we want to argue that for some $c\neq 0,1$, we can always find some root $\gamma$ for $\gamma^{p^k}+\gamma+(1-c^{-1})=0$ in $\F_{p^n}$, $p$ an odd prime.  Let $g=\gcd(n,k)$.

   We recall here the result from~\cite{CM04,Li78} (we simplify some parameters, though). Let $f(z)=z^{p^k}-A z-B$ in $\F_{p^n}$, $g=\gcd(n,k)$, $m=n/\gcd(n,k)$ and ${\rm Tr}_{\F_{p^n}/\F_{p^g}}$ be the relative trace from $\F_{p^n}$ to $\F_{p^g}$. For $0\leq i\leq m-1$, we define 
   $t_i=\frac{p^{n m}-p^{n(i+1)}}{p^n-1}$,
   $\alpha_0=A,\beta_0=B$. If $m>1$ (note that, if  $m=1$, gives $F(x)=x^2$, which was treated earlier), then, for $1\leq r\le m-1$, we let
 $
 \alpha_r=A^{\frac{p^{k(r+1)}-1}{p^k-1}} \text{ and } \beta_r=\sum_{i=0}^r A^{s_i} B^{p^{ki}},
 $
 where $s_i=\frac{p^{k(r+1)}-p^{k(i+1)}}{p^k-1}$, for $0\leq i\leq r-1$ and $s_r=0$. The trinomial $f$ has no roots in $\F_{p^n}$ if and only if $\alpha_{m-1}=1$ and $\beta_{m-1}\neq 0$. If  $\alpha_{m-1}\neq 1$, then it has a unique root, namely $x=\beta_{m-1}/(1-\alpha_{m-1})$, and, if $\alpha_{m-1}=1,\beta_{m-1}=0$, it has $p^g$ roots in $\F_{p^n}$ given by $x+\delta\tau$, where $\delta\in\F_{p^g}$, $\tau$ is fixed in $\F_{p^n}$ with $\tau^{p^k-1}=a$ (that is, a $(p^k-1)$-root of $a$), and, for any $e\in\F^*_{p^n}$ with ${\rm Tr}_g(e)\neq 0$, then
 $\displaystyle x=\frac{1}{{\rm Tr}_{\F_{p^n}/\F_{p^g}}(e)} \sum_{i=0}^{m-1} \left( \sum_{j=0}^i e^{p^{kj}}\right) A^{t_i} B^{p^{ki}}$.

Since in our case $A=-1$, $B=c+1$, and,
as we did in~\cite{EFRST20}, if $n/g$ is odd, then $\gamma^{p^k}+\gamma+(1-c^{-1})=0$ has a unique root in~$\F_{p^n}$, since $\alpha_{\frac{n}{g}-1}=-1$, independent of~$c$. 
Now, looking at $ x^{p^k}+ x+(c+1)=0$,
and by the same argument, it has a unique root, under $n/g$ odd, and so, 
  ${_c}\cB_F(a,b)\geq 2$ in this case (we use the prior root $x=0$, too).

Let $m:=n/g$ be even.  We switch the technique now. The equation in $x$ already has a root, namely $x=0$, so we want to show that there are values of $c$, for which $\gamma^{p^k}+\gamma+(1-c^{-1})=0$ has $p^g$ roots. We will, in fact, find many such classes of parameters $c$, below.

We denote by $\sigma_t\equiv t\pmod 2\in \{0,1\}$ the parity of $t$. Observe that $\alpha_0=\beta_0=-1$.
Since $\alpha_{m-1}=1$ for equation~$\gamma^{p^k}+\gamma+(1-c^{-1})=0$, we need to show that  $\beta_{m-1}=0$, to be able to use~\cite{CM04}. We compute (using the fact that the parities of $s_i$, $0\leq i\leq m-2$, are
$\sigma_{s_i}= (m-2-i+1)\pmod 2 = \sigma_{i-1}$, since $m$ is even),
\allowdisplaybreaks
\begin{align*}
\beta_{m-1}
&=\sum_{i=0}^{m-1} (-1)^{s_i}(1-c^{-1})^{p^{ki}}
=\sum_{i=0}^{m-1} (-1)^{i-1} \left(1-(-c^{-1})^{p^{ki}}\right)\\
&=\sum_{i=0}^{m-1} (-1)^{i-1} - \sum_{i=0}^{m-1} (-1)^{i-1+p^{ki}} (c^{-1})^{p^{ki}},\text{ since $m$ is even}\\
&=  \sum_{i=0}^{m-1} (-1)^{i-1} (c^{-1})^{p^{ki}}.
\end{align*}
Now, let $z\neq 1$ be such that $z^{p^{km}-1}-1=0$ (this always exists since $km$ is a multiple of $n$). Observe (we will be using that later) that $z^{p^{km}}-z=0$ and $z\neq 0,1$. 
We now set $c^{-1}=z^{p^k}+z$ and so, the previous displayed equation becomes
\allowdisplaybreaks
\begin{align*}
\beta_{m-1}&= \sum_{i=0}^{m-1} (-1)^{i-1} (z^{p^k}+z)^{p^{ki}}\\
&= \sum_{i=0}^{m-1} (-1)^{i-1} z^{p^{k(i+1)}}+  \sum_{i=0}^{m-1} (-1)^{i-1}  z^{p^{ki}}\\
&= \sum_{i=1}^{m-1} \left( (-1)^{i-1}+  (-1)^{i} \right) z^{p^{ki}} +z-z^{p^{km}}\\
&=z-z^{p^{km}}=0, \text{ from our choice of $z$}.
\end{align*}

 Therefore, by the result of~\cite{CM04,Li78}, we infer that the equation $z^{p^k}+z+(1-c^{-1})=0$ has $p^g$ solutions in $\F_{p^n}$. Thus, the initial Equation~\eqref{eq:x-gamma} has at least $p^g$ solutions, and so, the $c$-boomerang uniformity in this case (under $\frac{n}{\gcd{(n,k)}}$ even) is at least $p^g$, where $g=\gcd(n,k)$.

 Let us treat now the case of $d=(3^k+1)/2$  in $\F_{3^n}$, where the system is now  (recall that, since $p=3$, $1=-2$)
 \allowdisplaybreaks
 \begin{align*}
 &(x+\gamma)^{\frac{3^k+1}{2}}-c x^{\frac{3^k+1}{2}}=b\\
 &(x+\gamma+1)^{\frac{3^k+1}{2}}-c^{-1} (x+1)^{\frac{3^k+1}{2}}=b.
 \end{align*}
 We will not use the same method as in~\cite{CM04}, or~\cite{EFRST20} as permutation polynomials are not visibly involved here, rather we will modify the ``seed'' of the technique.
 Since $2|3^n-1$, for all $n$, then we can always write $x-1=y+y^{-1}$ and $x+\gamma-1=z+z^{-1}$, for some $y,z\in\F_{3^n}$. The system becomes
 \allowdisplaybreaks
   \begin{align*}
 &(z+z^{-1}-2)^{\frac{3^k+1}{2}}-c (y+y^{-1}-2)^{\frac{3^k+1}{2}}=b\\
 &(z+z^{-1}+2)^{\frac{3^k+1}{2}}-c^{-1} (y+y^{-1}+2)^{\frac{3^k+1}{2}}=b,
 \end{align*}
 that is,
 \allowdisplaybreaks
   \begin{align*}
 &\frac{(z-1)^{3^k+1}}{z^{\frac{3^k+1}{2}}}-c \frac{(y-1)^{3^k+1}}{y^{\frac{3^k+1}{2}}}=b\\
 &\frac{(z+1)^{3^k+1}}{z^{\frac{3^k+1}{2}}}-c^{-1} \frac{(y+1)^{3^k+1}}{y^{\frac{3^k+1}{2}}}=b,
 \end{align*}
 which, by expansion, renders
 \allowdisplaybreaks
    \begin{align*}
 &\frac{z^{3^k+1}-z^{3^k}-z+1}{z^{\frac{3^k+1}{2}}}-c \frac{ y^{3^k+1}-y^{3^k}-y+1 }{y^{\frac{3^k+1}{2}}}=b\\
 &\frac{z^{3^k+1}+z^{3^k}+z+1}{z^{\frac{3^k+1}{2}}}-c^{-1} \frac{ y^{3^k+1}+y^{3^k}+y+1 }{y^{\frac{3^k+1}{2}}}=b, \text{ or},
 \end{align*}
 \allowdisplaybreaks
 {\small
   \begin{align*}
 &T_{\frac{3^k+1}{2}}(x+\gamma-1)-T_{\frac{3^k-1}{2}}(x+\gamma-1)-cT_{\frac{3^k+1}{2}}(x-1)+cT_{\frac{3^k-1}{2}}(x-1)=b\\
 &T_{\frac{3^k+1}{2}}(x+\gamma-1)+T_{\frac{3^k-1}{2}}(x+\gamma-1)+c^{-1}T_{\frac{3^k+1}{2}}(x-1)+c^{-1}T_{\frac{3^k-1}{2}}(x-1)=b,
 \end{align*}
 }
  where $T_\ell(u+u^{-1})=u^\ell+u^{-\ell}$   is the Chebyshev polynomial of the first kind. Adding and subtracting these two equations will give
  \allowdisplaybreaks
{\small
     \begin{equation}
     \begin{split}
     \label{eq:T}
 &2T_{\frac{3^k+1}{2}}(x+\gamma-1)-(c-c^{-1})T_{\frac{3^k+1}{2}}(x-1)+(c+c^{-1})T_{\frac{3^k-1}{2}}(x-1)=2b\\
 &2T_{\frac{3^k-1}{2}}(x+\gamma-1)+(c+c^{-1})T_{\frac{3^k+1}{2}}(x-1)+(c^{-1}-c)T_{\frac{3^k-1}{2}}(x-1)=0.
 \end{split}
 \end{equation}
 }
  Thus, ${_c}\cB_F(a,b)\geq \cardinality\{(x,\gamma)\,|\, (x,\gamma) \text{ satisfies}~\eqref{eq:T}\}=\mu_c$.
  \end{proof}

\section{The $c$-boomerang uniformity for the   inverse function}
\label{sec7}

We now deal with the binary inverse function.

 \begin{thm}
 Let $n\geq 3$ be a positive integer, $0,1\neq c\in\F_{2^n}$ and $F:\F_{2^n}\to\F_{2^n}$ be the inverse function defined by $F(x)=x^{2^n-2}$. If $n=2$,   ${_c}\cB_F(a,b)\leq 1$ and if $n=3$,    ${_c}\cB_F(a,b)\leq 2$. If $n\geq 4$, ${_c}\cB_F(a,b)\leq 3$. Furthermore, ${_c}\cB_F(a,ab)=3$ \textup{(}so, the $c$-boomerang uniformity of $F$ is $\beta_{F,c}=3$\textup{)} if and only if any of the conditions happen:
\begin{enumerate}
\item[$(i)$] $\Tr(c)=0$ and there exists $b$ such that $(b^2 c+b   c^2+b+c^2+1)^2\neq 0$ and $\Tr\left(\frac{b^2c^2(bc+c+1) }{(b^2 c+b   c^2+b+c^2+1)^2}\right)=0$.
\item[$(ii)$] $\Tr(1/c)=0$ and there exists $b$ such that $(b^2 c+b   c^2+b+c^2+1)^2\neq 0$ and $\Tr\left(\frac{b^2c^2(bc+c+1) }{(b^2 c+b   c^2+b+c^2+1)^2}\right)=0$.
\item[$(iii)$] $\Tr\left(\frac{c^3}{(c^2+c+1)^2}\right)=0$ and there exists $b$ such that $(b^2 c+b   c^2+b+c^2+1)^2\neq 0$ and $\Tr\left(\frac{b^2c^2(bc+c+1) }{(b^2 c+b   c^2+b+c^2+1)^2}\right)=0$.
\item[$(iv)$] $c^2+c+1=0$ \textup{(}so, $n\equiv 0\pmod 2$\textup{)} and there exists $b$ such that $b^2 +b  +1\neq 0$ and $\Tr\left(\frac{b^2 c (b+c) }{(b^2 +b+1)^2}\right)=0$.
  \end{enumerate}
 \end{thm}
 \begin{proof}
 The claim about $n=2,3$ follows from the computation displayed in the appendix, so we now assume $n\geq 4$.
 
 By Theorem~\ref{thm:new-c-boom}, we need to investigate the system
 \begin{equation}
 \label{eq:inv2}
 \begin{array}{lcl}
 (x+\gamma)^{2^n-2}+c x^{2^n-2} &=&b \\
  (x+\gamma+a)^{2^n-2}+ c^{-1} (x+a)^{2^n-2}&=&b.
 \end{array}
  \end{equation}
 Observe that if $a=0$, then the system becomes $ (x+\gamma)^{2^n-2}+c x^{2^n-2}=b=(x+\gamma)^{2^n-2}+ c^{-1} x^{2^n-2}$, rendering $(c+c^{-1})x^{2^n-2}=0$, and since $c\neq 0,1$, then $x=0$. Thus, $\gamma^{2^n-2}=b$, which also has one solution $\gamma$, independent of~$b$, so ${_c}\cB_F(0,b)=1$.
 
 We next assume that $0\neq c\neq 1, a\neq 0$. Dividing \eqref{eq:inv2} by $a$ and relabeling 
 $x/a\mapsto x$, $\gamma/a\mapsto \gamma$ and $ab\mapsto b$ (all are linear equations), we are led to investigate the system
  \begin{equation}
 \label{eq:simpl-inv2}
 \begin{array}{lcl}
 (x+\gamma)^{2^n-2}+c x^{2^n-2} &=&b \\
  (x+\gamma+1)^{2^n-2}+ c^{-1} (x+1)^{2^n-2}&=&b.
 \end{array}
  \end{equation}
 
 If $b=0$, the system~\eqref{eq:simpl-inv2} transforms into $(x+\gamma)^{2^n-2}=c x^{2^n-2}, (x+\gamma+1)^{2^n-2}=c^{-1} (x+1)^{2^n-2}$. If $x=0$, then $\gamma=0$, which renders the non-permissible $c=1$. Similarly, $x=1$, $x=\gamma$, $x=\gamma+1$ can only happen if $c=1$ (assuming $b=0$) and that is not allowed. If none of these values of $x$ happen then the system becomes
 \[
 x+\gamma=c^{-1} x,\ x+\gamma+1=cx+c,
 \]
 with solutions $\gamma=1,x=\frac{c}{c+1}\neq 0,1$, so $x\notin\{0,1,\gamma,\gamma+1\}$, hence ${_c}\cB_F(a,0)=1$ (which holds for all $a$, using the previous discussion). We next assume that~$b\neq 0$.
 
 \vskip.2cm
 
 \noindent
{\em Case $1$}:
 $b\neq 0$, $x=0$. The system~\eqref{eq:simpl-inv2} reduces to $\gamma^{2^n-2}=b$ and $(\gamma+1)^{2^n-2}- c^{-1} =b$. Since $b\neq 0$, then $\gamma=1/b$, which can only happen if $\frac{b}{b+1}=b+c^{-1}$, that is,  $b^2+c^{-1}b+c^{-1}=0$. This equation in $b$ has two solutions, say, $b_1,b_1'=1/(cb_1)$, if and only if, by Lemma~\ref{lem10}, $\Tr(c^{-1}/c^{-2})=\Tr(c)=0$. Thus, in this case, we have a contribution of~$1$ (the solutions are $(x,\gamma)=(0,1/b_1),(0,1/b_1')$) to ${_c}\cB_F(a,ab_1 )$, respectively, $\cB_F(a,ab_1')$, otherwise there is no contribution.

\vskip.2cm

\noindent
{\em Case $2$}: $b\neq 0$, $x=1$. The system~\eqref{eq:simpl-inv2} reduces to
$(\gamma+1)^{2^n-2}=b+c, \gamma^{2^n-2}=b$. If $b=1$, then $\gamma=1$ and $c$ must be~1, an impossibility. If $b=c\neq 1$, then $\gamma=1/b$ and $(\gamma+1)^{2^n-2}=0$, so, $\gamma=1$, implying $b=1$, and so, $c=1$, which is not allowed. Thus, $1\neq b\neq c$, and so, $\gamma=1/b$ renders $b^2+cb+c=0$, which has two roots, say $b_2,b_2'=c/b_2$, if and only if, by Lemma~\ref{lem10}, $\Tr(1/c)=0$. Therefore, if $\Tr(1/c)=0$,  
 then we have a contribution of~$1$ to ${_c}\cB_F(a,ab_2)$, respectively, ${_c}\cB_F(a,ab_2')$ (the solutions are $(x,\gamma)=(1,1/b_2),(1,1/b_2')$). 
 
 Can $b_1,b_1'$ equal $b_2,b_2'$? If that is so, then $c^{-1}b_1+c^{-1}=cb_1+c$, obtaining $b_1=1$, but that is impossible since it will not satisfy $b^2+bc+c=0$ (similarly for $b_1'=b_2$, etc.). Thus, the two contributions from Case~$1$ and~$2$ will not overlap.
 
 \vskip.2cm
 
 \noindent
{\em Case $3$}: $b\neq 0$, $x=\gamma$. 
  The system~\eqref{eq:simpl-inv2} reduces to
 $c x^{2^n-2}=b, 1+c^{-1} (x+1)^{2^n-2}=b$, so $x=c/b$. If $b=c$, then $x=1$, and $b=1$ (from the second equation), an impossibility. If $b\neq c$, then we must have $b^2+\frac{c^2+c+1}{c}b+c=0$, and two roots  $b$ exist ($b_3, c/b_3$), under $c^2+c+1\neq 0$,  by Lemma~\ref{lem10}, if and only if $\Tr\left(\frac{c^3}{(c^2+c+1)^2}\right)=0$ (one such $b_3'$ does exist even  when $c^2+c+1=0$ (so, $n$ is even), as then, $b_3'=c^{1/2}$, which always exists). Under these circumstances,
 we have a contribution of~$1$ to ${_c}\cB_F(a,ab_3)$, ${_c}\cB_F(a,ac/b_3)$, respectively, ${_c}\cB_F(a,ab_3')$. 
 
 Can any of $b_3$, or $b_3'$ be equal to $b_1$, or $b_2$? 
  \begin{itemize}
 \item 
 If $b_3=b_1$, then as above we get that $b_1=b_3=\frac{c^2+c+1}{(c+1)^2}$, assuming $\Tr(c)=\Tr\left(\frac{c}{c^2+c+1}\right)=0$. If that is so, plugging this value in any of the component equations, renders $\frac{c^4}{(c+1)^4}=0$, an impossibility.
 \item  If $b_3=b_2$, similarly, we get that $b=\frac{c}{c+1}$, assuming $\Tr\left(\frac{1}{c}\right)=0$ and $\Tr\left(\frac{c}{c^2+c+1}\right)=0$. Putting  the value of $b$ in any of the component equations, gives $\frac{c}{(c+1)^2}=0$, which is impossible.  
 \item  $b_3'=c^{1/2}=b_1$ (under $c^2+c+1=0$) when plugged into the equation of $b_1$, gives $c=1$, a contradiction.
 \item $b_3'=c^{1/2}=b_2$, similarly renders (under $c^2+c+1=0$) when put into the equation of $b_2$, the value $c=0$, a contradiction.
 \end{itemize}
 
  \noindent
{\em Case $4$}: $b\neq 0$, $x=\gamma+1$. 
The system~\eqref{eq:simpl-inv2} reduces to $1+c x^{2^n-2}=b, c^{-1} (x+1)=b$. If $b=1$, then $x=0$, and so, $c^{-1}=b=1$, a contradiction. Thus, $b\neq 1$ and  $x=c/(b+1)$, which, when used in the first equation, gives  
\[
b^2+\frac{c^2+c+1}{c} b+\frac{1}{c}=0,
\]
which renders two solutions, say  $b_4,b_4'=1/(cb_4)$, if and only if $\Tr\left( \frac{c}{(c^2+c+1)^2}\right)=0$. For such a $c$, we get a contribution of   $1$ to ${_c}\cB_F(a,ab_4)$, respectively,  ${_c}\cB_F(a,ab_4')$, assuming $\Tr\left( \frac{c}{(c^2+c+1)^2}\right)=0$.

 Can $b_4$ (or $b_4'$) be equal to $b_1,b_1'$, $b_2,b_2'$,   $b_3$, or $b_3'$? 
 \begin{itemize}
 \item If $b_4=b_1$ (or $b_1'$), then they must satisfy both
 $cb^2+b+1=0, cb^2+(c^2+c+1) b+1=0$, from which we infer $b=0$ or $c=0,1$, an impossibility.
\item  If $b_4=b_2$ (or $b_2'$), then they must satisfy both $b^2+bc+c=0,cb^2+(c^2+c+1) b+1=0$, 
implying $b=c+1$, which does not satisfy $b^2+bc+c=0$.

\item  If $b_4=b_3$, then they must satisfy both equations $c b^2+(c^2+c+1)b+c^2=0,c b^2+(c^2+c+1) b+1=0$, inferring $c=1$, an impossibility.
 
 \item If $b_4=b_3'(=c^{1/2})$, then $c^2+c+1=0$ (so, $n\equiv 0\pmod 2$) and $cb_4^2+(c+1)b+(c+1)=0$. Thus, $c^2=(c+1)c^{1/2}$, which combined with $c^2+c+1=0$, renders $(c+1)(c^{1/2}+1)=0$, that is, $c=1$, an impossibility.
  \end{itemize}

\vspace{.2cm}

\noindent
 {\em Case $5$.}
 Assuming now that $x(x+1)(x+\gamma)(x+\gamma+1)\neq 0$, and multiplying the first equation of~\eqref{eq:simpl-inv2} by $x(x+\gamma)$ and the second by $c(x+1)(x+\gamma+1)$ we obtain the new system
 \begin{equation}
 \begin{split}
 \label{eq:simpl-inv3}
1+c x + b x^2 + (c + b x) \gamma&=0\\
1 + c + b c + (1 + c) x + b c x^2 + (1 + b c + b c x) \gamma&=0.
 \end{split}
 \end{equation}
 Observe that $x\neq c/b$ since, otherwise, $c=0$ (inferred from the first equation), a contradiction. If $x\neq c/b$, then using $\displaystyle \gamma=\frac{b x^2+c x+x}{b x+c}$ found from the first equation and replacing it into the second, we get 
 \[
b^2 c x^2+ (b^2 c+b   c^2+b+c^2+1) x +  b c^2 +c^2+c=0,
 \]
 or
 \begin{equation}
 \label{eq:simpl-inv4}
 x^2+\frac{ b^2 c+b   c^2+b+c^2+1}{b^2 c} x+\frac{b   c+c+1}{b^2}=0.
 \end{equation} 
 If $b^2 c+b   c^2+b+c^2+1=0$, we have a contribution of~$1$ to ${_c}\cB_F(a,ab)$.
 If  $b^2 c+b   c^2+b+c^2+1\neq 0$, by Lemma~\ref{lem10}, Equation~\eqref{eq:simpl-inv3}  will have two solutions if and only if 
 \begin{equation}
 \label{eq:simpl-inv5}
  \Tr\left(\frac{\frac{b   c+c+1}{b^2}}{\left( \frac{b^2 c+b   c^2+b+c^2+1}{b^2 c}\right)^2} \right)=\Tr\left(\frac{b^2c^2(bc+c+1) }{(b^2 c+b   c^2+b+c^2+1)^2}\right)=0.
 \end{equation} 
 Therefore, under the above trace condition on $b$ for a fixed $c$, we have a contribution of~$2$ to  ${_c}\cB_F(a,ab)$.
 
 To avoid speaking about an empty condition, we want to argue next that~\eqref{eq:simpl-inv5} will happen for a fixed $c\neq 0,1$, so it will be sufficient to find a values of $b$ where the trace above is 0. 
 If $b=1$ (and so, $1 + b + b^2 c + c^2 + b c^2\neq 0$), then the trace becomes $\Tr(1)=0$, which will always happen if the dimension is even, the Equation~\eqref{eq:simpl-inv4} is then $x^2+x+1=0$ (with two distinct roots),  so in this case we have a contribution of $2$ to  ${_c}\cB_F(a,ab)$, where $b$ satisfies~\eqref{eq:simpl-inv5}, otherwise there is none. 
 
 For the remainder of the proof below we assume that $n$ is odd.
 Taking, for example, $b=c^2+c^{-1}\neq 0$ (for this value, $1 + b + b^2 c + c^2 + b c^2\neq 0$, since otherwise, we would get $c^5+c^4+c+1=(c+1)^5=0$, so $c=1$, a contradiction), since the dimension is odd and $c\neq 0,1$, the 
 above trace calculation reduces to 
 \[
 \Tr\left(\frac{c(1+c+c^2)^2}{(1+c)^4} \right)=\Tr\left(\frac{c}{(1+c)^2}+\frac{c^2}{(1+c)^4} \right)=0,
 \]
 which is obviously true, and so there are two roots $x$ to the Equation~\eqref{eq:simpl-inv4}.

  We need to make sure that the solutions $x,\gamma$ will not  satisfy  $x(x+1)(x+\gamma)(x+\gamma+1)=0$, so we go back (replacing $b=c^2+c^{-1}$) to Equation~\eqref{eq:simpl-inv4}, obtaining
 \[
x^2+ \frac{c^4+c^3+c^2+c  }{(c^2+c+1)^2}x+\frac{c^3}{(c^2+c+1)^2}=0
 \]
 If $x=0$, then $c=0$, which is not allowed. If $x=1$, then we get  $\frac{1+c}{(c^2+c+1)^2}=0$, and so, $c=1$, which is not allowed. If $x=\gamma$, then the system~\eqref{eq:simpl-inv3} becomes 
 \[
x=0, (c^3+c+1)x+c^3+c=0,
 \]
which is impossible since $c\neq 0,1$. If $x=\gamma+1$, the system is now
 \[
 x+1=0, (c^3+c+1)x+c^2=0,
 \]
 which can only happen if $c^3+c^2+c+1=(c+1)^3=0$, an impossibility.

 Putting together the above cases, we see that, if $n\geq 4$,  the only possibility is for at most three solutions of the $c$-boomerang system:
 we get three solutions if Case 5 is combined with any of the previous 4 cases (for $(iv)$ we use $c^2+c=1=0$ to simplify the trace expression), and therefore,
 the $c$-boomerang uniformity is~$\leq 3$, but not higher. Our theorem is shown.
 \end{proof}

\begin{remark}
From the previous proof, we do get a lot more information about the Boomerang Connectivity Table for the binary inverse function, but we preferred to simply give just the maximum $c$-BCT entries.  \end{remark}

 We now treat the case of the inverse for odd characteristic. We let $[A]^2=\{x^2\,|\, x\in A \}$, where $A$ is a set with a defined multiplication on it.
 \begin{thm}
 Let $p$ be an odd prime, $n\geq 1$ be a positive integer, $0,1\neq c\in\F_{p^n}$ and $F:\F_{p^n}\to\F_{p^n}$ be the inverse $p$-ary function defined by $F(x)=x^{p^n-2}$. 
 For any $a,b\in\F_{p^n}$,  the $c$-BCT entries   ${_c}\cB_F (a,ab)\leq 4$. Furthermore,
 \begin{enumerate}
 \item[$(i)$]  If $1 + 2 c - 2 c^2 + 2 c^3 + c^4=0$ and $3 - 2 c + 3 c^2, 1-4c, c^2-4c\in  [\F_{p^n}]^2$, then $\beta_{F,c}=4$ \textup{(}e.g., for such  $c$,  if  $b=\frac{c^2-1}{c^2+1}$, then ${_c}\cB_F (a,ab)=4$\textup{)}.
   \item[$(ii)$]   If $1 - 2 c - 2 c^2 - 2 c^3 +c^4=0$ and $1 - 6 c + c^2, 1-4c, c^2-4c\in  [\F_{p^n}]^2$, then    $\beta_{F,c}=4$ \textup{(}e.g., for such  $c$,  if  $b=\frac{c^2-1}{2c}$, then ${_c}\cB_F (a,ab)=4$\textup{)}.
    \end{enumerate}
 \end{thm}
 \begin{proof}
 For $a=0$,  the corresponding system $(x+\gamma)^{p^n-2}-cx^{p^n-2}=b, (x+\gamma)^{p^n-2}-c^{-1} x^{p^n-2}=b$ has one solution $(x,\gamma)$. 
 We assume now that $a\neq 0$. Multiplying by $a$ and relabeling,   the $c$-boomerang system becomes 
\begin{equation}
\begin{split}
\label{eq:inverse2}
 (x+\gamma)^{p^n-2}-cx^{p^n-2}=b.\\
 (x+\gamma+1)^{p^n-2}-c^{-1} (x+1)^{p^n-2}=b.
 \end{split}
\end{equation}
 If $b=0$, we easily get only one solution, so we may assume below that $b\neq 0$.

\vspace{.2cm}
\noindent
{\em Case $1$.} Let $x=0$. The $c$-boomerang system is $\gamma^{p^n-2}=b, (\gamma+1)^{p^n-2}-c^{-1} =b$, so $\gamma=1/b$, and $(\frac{1}{b}+1)^{p^n-2}=b+\frac{1}{c}$ (surely, $b+\frac{1}{c}\neq 0$, since, otherwise, $b=-1$, and so, $c=1$, an impossibility). Thus, $(0,\frac{1}{b})$ is a solution assuming $b$ satisfies $b^2+\frac{1}{c } b+\frac{1}{c}=0$. By Lemma~\ref{lem10}, a unique $b_1'$ exists if and only if the discriminant $D_1=c^{-2}-4 c^{-1}=0$, that is $c=4^{-1}$. Again, from Lemma~\ref{lem10}, two solutions $b_1$ (we call these, here and in the next three cases, by the same label, as it will not matter in our argument) exist if and only if $1-4c\in [\F_{p^n}]^2$, $c\neq 4^{-1}$. In either case, there is a contribution of~$1$ to the respective $c$-BCT entry.

\vspace{.2cm}
 \noindent
{\em Case $2$.} 
If $x=-1$, the $c$-boomerang system is now $(\gamma-1)^{p^n-2}+c=b, \gamma^{p^n-2}=b$, so $\gamma=1/b$, which used in the first equation renders $(\frac{1}{b}-1)^{p^n-2}=b-c$, which simplified gives $b^2-bc+c=0$ (again, $b\neq c$, since otherwise $1/b=1$, and so, $c=1$, an impossibility). By  Lemma~\ref{lem10}, a unique $b_2'$ exists if and only if the discriminant $D_2=c^2-4c=0$, so $c=4$, and two solutions $b_2$ exist if and only if $D_2\in  [\F_{p^n}]^2$, $c\neq 4$ (thus, $D_2\in  [\F_{p^n}^*]^2$). In either case, there is a contribution of~$1$ to the respective $c$-BCT entry.

\vspace{.2cm}
\noindent
{\em Case $3$.} Let $x=-\gamma$. The $c$-boomerang system is now
$c\gamma^{p^n-2}=b, 1-c^{-1}(1-\gamma)^{p^n-2}=b$. Thus, $\gamma=c/b$ (if $b=0$, then $\gamma=0$, and the second equation gives us, $c^{-1}=b=0$, an impossibility). Using $\gamma=c/b$ (observe that $b\neq 1$)  in the second equation we obtain  $b^2-b \frac{c^2+c-1}{c}+c=0$, which has a unique solution $b_3'$ if and only if $D_3=\frac{1 - 2 c - c^2 - 2 c^3 + c^4}{c^2}=\frac{(1 - 3 c + c^2) (1 + c + c^2)}{c^2}=0$. There are two roots $b_3$ if and only if $D_3 \in [\F_{p^n}^*]^2$, or equivalently, $(1 - 3 c + c^2) (1 + c + c^2)\in   [\F_{p^n}^*]^2$.

\vspace{.2cm}
\noindent
{\em Case $4$.} Let $x=-\gamma-1$. The $c$-boomerang system becomes
$-1-c x^{p^n-2}=b, -c^{-1}(x+1)^{p^n-2}=b$.  Then $x=-\frac{c}{b+1}$ (if $b=-1$, then $x=0$, and using the second equation, we get $-c^{-1}=b$, that is, $c=1$, an impossibility), which used in the second equation gives $b^2-b\frac{c^2-c-1}{c}+\frac{1}{c}=0$. This equation has a root $b_4'$ if and only if $D_4=\frac{1 - 2 c - c^2 - 2 c^3 + c^4}{c^2}=\frac{(1 - 3 c + c^2) (1 + c + c^2)}{c^2}=0$, and two roots $b_4$ if and only if $D_4\in [\F_{p^n}^*]^2$, or equivalently, $(1 - 3 c + c^2) (1 + c + c^2)\in [\F_{p^n}^*]^2$.

 We now need to check if overlaps exist among the $b$'s of various cases, and therefore the contributions to the $c$-BCT entries will be added. 
 \begin{itemize}
\item If $b_1$ (or $b_1'$) equals $b_2$ (or $b_2'$), then these must satisfy both 
 $b^2+\frac{1}{c} b+\frac{1}{c}=0$, $b^2-bc+c=0$, that is, $b=\frac{c^2-1}{c^2+1}$, which when used in the first equation gives $1 + 2 c - 2 c^2 + 2 c^3 + c^4=0$. Observe that $c=4,4^{-1}$ vanish this expression if and only if $p=19$. So,  we have a contribution of~$2$ to the respective $c$-BCT entry if and only if $1 + 2 c - 2 c^2 + 2 c^3 + c^4=0$ (this includes the cases $c=4,4^{-1}$).
 
\item If $b_1$ (or $b_1'$) equals $b_3$ (or $b_3'$), then these must satisfy both $b^2+\frac{1}{c} b+\frac{1}{c}=0$, $b^2-\frac{c^2+c-1}{c} b+c=0$, giving $b=\frac{c-1}{c}$, which when used in the first equation implies $1=0$, an impossibility.
 
\item  If $b_1$ (or $b_1'$) equals $b_4$ (or $b_4'$), then these must satisfy both $b^2+\frac{1}{c} b+\frac{1}{c}=0$, $b^2-\frac{c^2-c-1}{c} b+\frac{1}{c}=0$, implying that either $c=0,1$, or $b=0$, none of which will work.
  
  \item  If $b_2$ (or $b_2'$) equals $b_3$ (or $b_3'$), then these must satisfy both $b^2-bc+c=0$, $b^2-\frac{c^2+c-1}{c} b+c=0$, which will not work (same argument as in the previous item).
 
   \item  If $b_2$ (or $b_2'$) equals $b_4$ (or $b_4'$), then these must satisfy both $b^2-bc+c=0$, 
  $b^2-\frac{c^2-c-1}{c} b+\frac{1}{c}=0$, so $b=c-1$, which when replaced into the first equation gives $1=0$,  an impossibility.
  
   \item  If $b_3$ (or $b_3'$) equals $b_4$ (or $b_4'$), then these must satisfy both  $b^2-\frac{c^2+c-1}{c} b+c=0$, $b^2-\frac{c^2-c-1}{c} b+\frac{1}{c}=0$, and so, $b=\frac{c^2-1}{2c}$, which implies~$1 - 2 c - 2 c^2 - 2 c^3 +c^4=0$.
     
\end{itemize}  
Observe that other combinations cannot occur above.

\noindent
{\em Case $5$.}  Assume now that $x\neq 0,1,\gamma,\gamma+1$.   Multiplying the first equation by $x(x+\gamma)$ and the second by $c(x+1)(x+\gamma+1)$ renders
\begin{align*}
x - c (x + \gamma) &= b x (x + \gamma)\\
c (x + 1) - (x + \gamma + 1) &= b c (x + 1) (x + \gamma + 1).
\end{align*}
Solving for $\gamma$ in the first equation, we get $\gamma=\frac{x - c x - b x^2}{c + b x}$ (observe that $x\neq -c/b$), which used in the second equation gives us the equation
\[
x^2+\frac{bc^2+b^2c+b+1-c^2}{b^2c} x+\frac{bc-c+1}{b^2}=0,
\]
with a unique root $x$ if  and only if 
\begin{align*}
b^4 c^2 D_5 &= 1 + 2 b + b^2 + 2 b^2 c + 2 b^3 c - 2 c^2 - 2 b^2 c^2\\
&\quad \quad + 
  b^4 c^2 + 2 b^2 c^3 - 2 b^3 c^3 + c^4 - 2 b c^4 + b^2 c^4\\
  & =(1 + b - c) (1 - c + b c) (1 + b + 2 c + b^2 c + c^2 - b c^2)=0,
  \end{align*} and two distinct roots $x$ if and only if
   \[
  (1 + b - c) (1 - c + b c) (1 + b + 2 c + b^2 c + c^2 - b c^2)\in [\F_{p^n}^*]^2.
  \]

   Putting together our discussion, we see that the largest $c$-BCT entry can only be $4$; we do  get that value if  Case 5 is combined with combinations of the other cases $c$-BCT namely, both Case 1 and Case 2, or both Case 3 and Case 4.
 
  The proof of the theorem is done.
  \end{proof}
  
  \begin{remark}
From the previous proof, we do get a lot more information about the $c$-Boomerang Connectivity Table for the $p$-ary inverse function, but,  as for the binary case,  we preferred to just give the maximum $c$-BCT entries.  
\end{remark}

\section{Concluding remarks}
\label{sec8}

We defined a new concept, we call  $c$-boomerang uniformity based upon a previously defined multiplicative differential. We characterized the new concept in terms of the Walsh transforms and investigated the properties of some perfect nonlinear functions, as well as the inverse function in all characteristics via this new concept.

It would certainly be interesting to see how other perfect nonlinear, as well as almost perfect nonlinear functions behave under this new $c$-boomerang uniformity.

\section{Appendix}

 We list all the uniformity values for the considered functions in the paper, for $2\leq n\leq 4$, and some values of the parameter $k$, if involved, computed via SageMath. Below, $\alpha$ denotes a primitive root in the respective field. For $p=2$, to construct $\F_{2^n}$, $2\leq n\leq 4$, we take the primitive polynomials, $x^2 + x + 1,x^3 + x + 1,x^4+x+1$, respectively. If $p=3$, to construct $\F_{3^n}$, $2\leq n\leq 5$, we take the primitive polynomials, $x^2 + 2x + 2, x^3 + 2x + 1, x^4 + 2x^3 + 2, x^5 + 2x + 1$, respectively.   If $p=5$, to construct $\F_{5^n}$, $2\leq n\leq 4$, we take the primitive polynomials, $x^2 + 4 x + 2, x^3 + 3 x + 3, x^4 + 4 x^2 + 4 x + 2$.

We will only mention the nonzero entries (surely, for some parameters, $a,b,c$, ${_c}\cB_F(a,b)=0$, but we are not concerned with that). First, if $f(x)=x^2$, the possible $c$-BCT entries $(c\neq 0,1$)for $f(x)=x^2$, $2\leq n\leq 4$,  are $[1, 2]$ for $n=2$, respectively, $ [1, 2, 3, 4] $ for $n=3,4$.
 

\captionsetup{width=4.5cm}
\begin{table}[H]
\floatbox[\capbeside]{table}
{\caption{Possible $c$-BCT entries for the inverse $f(x)=x^{p^n-2}$, $2\leq n\leq 4$, $p=2, 3$,  $c\neq 0,1$}}
{\flushleft\begin{tabular}{|c|c|l|}
\hline
$p$ & $n$ & possible values of ${_c}\cB_F(a,b)$  \\
\hline
\multirow{3}{*}{2} & 2 & $[1]$\\ 
      & 3 & $[1, 2]$\\ 
      & 4  & $[1, 2, 3]$\\
      \hline
  \multirow{3}{*}{3} & 2 & $[1, 2]$\\ 
      & 3 & $[1, 2, 3]$\\ 
      & 4 & $[1, 2, 3, 4]$\\
\hline
\end{tabular}}
\end{table}



\captionsetup{width=4.2cm}

\begin{table}[H]
\floatbox[\capbeside]{table}
{\caption{Possible $c$-BCT entries for the Gold $f(x)=x^{3^k+1}$, $2\leq n\leq 4$, $1\leq k\leq 3$, $c\neq 0,1$}}
{\flushleft\begin{tabular}{|c|c|l|}
\hline
$k$ & $n$ & possible values of ${_c}\cB_F(a,b)$  \\
\hline
\multirow{3}{*}{1} & 2 & $[1, 2, 4, 6, 9]$\\ 
      & 3 & $[1, 2, 3]$\\ 
      & 4  & $[1, 2, 3, 4, 5, 6, 7, 8, 9]$\\
      \hline
  \multirow{3}{*}{2} & 2 & $[1, 2]$\\ 
      & 3 & $[1, 2, 3]$\\ 
      & 4  & $[1, 2, 4, 72, 73, 81, 82, 90, 91, 100]$\\
\hline
  \multirow{3}{*}{3} & 2 & $[1, 2, 4, 6, 9]$\\ 
      & 3 & $[1, 2, 3, 4]$\\ 
      & 4  & $[1, 2, 3, 4, 5, 6, 7, 8, 9]$\\
\hline
\end{tabular}}
\end{table}


\captionsetup{labelfont=bf,singlelinecheck=false, position=above}
\captionsetup{width=4.2cm,position=top}

\begin{table}[H]
\floatbox[\capbeside]{table}
{\caption{Possible $c$-BCT entries for $f(x)=x^{\frac{3^k+1}{2}}$, $2\leq n\leq 4$, $1\leq k\leq 3$, $c\neq 0,1$}}
{\flushleft\begin{tabular}{|c|c|l|}
\hline
$k$ & $n$ & possible values of ${_c}\cB_F(a,b)$  \\
\hline
\multirow{3}{*}{1} & 2 & $[1, 2]$\\ 
      & 3 & $[1, 2, 3, 4]$\\ 
      & 4  & $[1, 2, 3, 4]$\\
      \hline
  \multirow{3}{*}{2} & 2 & $[1, 2]$\\ 
      & 3 & $[1, 2, 3, 4]$\\ 
      & 4  & $[1, 2, 3, 4, 5, 6, 7, 8, 9, 10, 12, 18]$\\
\hline
  \multirow{3}{*}{3} & 2 & $[1, 2]$\\ 
      & 3 & $[1, 2, 3, 4, 6, 7, 8]$\\ 
      & 4  & $[1, 2, 3, 4, 5]$\\
\hline
\end{tabular}}
\end{table}


  Let $f(x)=x^{10}-x^6-x^2$ on $\F_{3^2}$. We list all values of $(c,b)$, for which the $c$-DDT entry  (for some $a$) equals the  $c$-boomerang uniformity of~$2$ (all  $c$-BCT entries of $[1, 2]$ occur):
 $(2, 2 \alpha), (\alpha + 2, \alpha + 2), (2 \alpha, 2 \alpha), (\alpha + 2, 2 \alpha), (2 \alpha + 1, 2 \alpha), 
 (2 \alpha + 1, \alpha + 2), (\alpha, 2 \alpha), (2, 2), (2, \alpha + 2), (2 \alpha, \alpha + 2), (\alpha, \alpha + 2).
$
 
   Let $f(x)=x^{10}-x^6-x^2$ on $\F_{3^3}$. We list all values of $(c,b)$, for which the $c$-DDT entry  (for some $a$) equals the  $c$-boomerang uniformity of~$4$ (all   $c$-BCT entries of $[1, 2, 3, 4]$ occur):
$ 
(2 \alpha, \alpha^2), (\alpha^2 + \alpha, \alpha^2 + \alpha + 1), (\alpha^2 + 2, \alpha^2), (2 \alpha + 1, \alpha^2 + \alpha + 1), 
 (2 \alpha + 2, \alpha^2 + 2 \alpha + 1), (\alpha^2 + 2 \alpha, \alpha^2 + 2 \alpha + 1).$
 
   Let $f(x)=x^{10}-x^6-x^2$ on $\F_{3^4}$. We list all values of $(c,b)$, for which the $c$-DDT entry  (for some $a$) equals the  $c$-boomerang uniformity of~$6$ (all   $c$-BCT entries of $[1, 2, 3, 4, 5, 6]$ occur):
$(2, 2 \alpha^3 + 2 \alpha^2), (2, \alpha^3 + \alpha^2 + 2).$
 
  Let $f(x)=x^{10}-x^6-x^2$ on $\F_{3^5}$ (the complexity of this computation is about $2^{32}$, and if we were to go up to $n=6$, it would be $2^{40}$ operations, so we stopped at $n=5$). 
 We list all values of $(c,b)$, for which the $c$-DDT entry  (for some $a$) equals the  $c$-boomerang uniformity of~$6$ (all   $c$-BCT entries of  $[1, 2, 3, 4, 5, 6]$ occur):
 \allowdisplaybreaks
{ \scriptsize 
 \begin{align*}
&(\alpha^4 + \alpha^3 + \alpha^2 + 2 \alpha, \alpha^4 + \alpha^2 + 2 \alpha + 1), (2 \alpha^2 + 2 \alpha, \alpha^4 + 2 \alpha^3 + \alpha^2 + 1),\\
& (\alpha^4 + \alpha^2 + 2, 2 \alpha^3 + 2 \alpha),  (2 \alpha^4 + \alpha^3 + \alpha^2 + 2 \alpha + 2, \alpha^3 + \alpha^2 + \alpha), \\
& (2 \alpha^3 + \alpha^2 + \alpha + 2, \alpha^4 + 2 \alpha^3 + 2 \alpha + 1), (2 \alpha + 1, 2 \alpha^4 + \alpha^2 + 2 \alpha + 2), \\
& (\alpha^3 + \alpha^2 + \alpha, \alpha^4 + 2 \alpha^3 + 1), (\alpha^4 + \alpha^3 + 2 \alpha^2 + 2 \alpha + 2, \alpha^4 + 2 \alpha^3 + 2 \alpha + 1), \\
& (2 \alpha^4 + \alpha^3 + \alpha^2 + 2, 2 \alpha^4 + \alpha^3 + \alpha^2 + 2 \alpha + 2), (\alpha^4 + 2 \alpha^2, \alpha^4 + \alpha^3 + 2 \alpha^2 + \alpha + 1),\\
&  (2 \alpha^4 + \alpha^3 + 2 \alpha^2 + 2 \alpha + 1, 2 \alpha^4 + \alpha^3 + 2 \alpha^2 + \alpha + 2), (2 \alpha^2 + 2, \alpha^4 + 2 \alpha^2 + 2 \alpha + 1), \\
& (\alpha^4 + \alpha^3 + 1, 2 \alpha^4 + \alpha^3 + \alpha^2 + 2), (\alpha^3 + 2 \alpha, \alpha^4 + 2 \alpha^2 + 2 \alpha + 1), \\
& (\alpha^4 + 2 \alpha + 2, \alpha^4 + \alpha^3 + 2 \alpha^2 + \alpha + 1), (\alpha^4 + \alpha^3 + \alpha^2 + \alpha, 2 \alpha^4 + \alpha^2 + 2 \alpha + 2), \\
& (2 \alpha^3 + 1, 2 \alpha^4 + \alpha^3 + 2 \alpha + 2), (2 \alpha^4 + 2 \alpha^3 + \alpha + 2, \alpha^4 + \alpha + 1),\\
&  (2 \alpha^3 + \alpha + 2, 2 \alpha^4 + \alpha^3 + \alpha^2 + 2), (\alpha^4 + \alpha^3 + 2 \alpha^2 + 1, \alpha^4 + \alpha^2 + 2 \alpha + 1),\\
&  (2 \alpha^4 + 2 \alpha^3 + \alpha^2 + 1, \alpha^3 + \alpha^2 + \alpha), (2 \alpha^3 + 2 \alpha^2 + 2, 2 \alpha^4 + \alpha^3 + 2 \alpha + 2),\\
&  (\alpha^4 + 2 \alpha^3 + 1, \alpha^4 + \alpha^3 + 1), (\alpha^3 + \alpha^2 + 2 \alpha + 2, \alpha^4 + \alpha^3 + 1),\\
& (\alpha^3 + 2 \alpha^2 + \alpha + 2, \alpha^4 + 2 \alpha^3 + \alpha^2 + 1), (\alpha^4 + \alpha^3 + \alpha, 2 \alpha^3 + 2 \alpha), \\
& (2 \alpha^2 + \alpha + 2, \alpha^4 + \alpha + 1), (2 \alpha^3 + 2 \alpha^2 + \alpha, 2 \alpha^4 + \alpha^3 + 2 \alpha^2 + \alpha + 2),\\
& (\alpha^2 + \alpha + 2, 2 \alpha^4 + \alpha^3 + \alpha^2 + 2 \alpha + 2), (\alpha^4 + \alpha^3 + 2 \alpha, \alpha^4 + 2 \alpha^3 + 1).
 \end{align*}
 }

   Let $f(x)=x^{10}+x^6-x^2$ on $\F_{3^2}$. We list all values of $(c,b)$, for which the $c$-DDT entry  (for some $a$) equals the  $c$-boomerang uniformity of~$2$ (all  $c$-BCT entries of $[1, 2]$ occur):
$
 (2, 2 \alpha + 1), (2 \alpha + 1, \alpha), (2 \alpha, 2 \alpha + 1), (\alpha, \alpha), (\alpha + 2, 2 \alpha + 1), 
 (2 \alpha + 1, 2 \alpha + 1), (2, 1), (2 \alpha, \alpha), (2, \alpha), (\alpha, 2 \alpha + 1), (\alpha + 2, \alpha). 
 $
 
   Let $f(x)=x^{10}+x^6-x^2$ on $\F_{3^3}$. We list all values of $(c,b)$, for which the $c$-DDT entry  (for some $a$) equals the  $c$-boomerang uniformity of~$3$ (all  $c$-BCT entries of $[1, 2, 3]$ occur):
   \allowdisplaybreaks
{  \scriptsize 
   \begin{align*}
& (\alpha^2 + \alpha + 2, \alpha^2 + 2 \alpha + 2), (2 \alpha^2 + 2 \alpha, \alpha^2 + 2 \alpha + 2), (\alpha^2 + \alpha + 2, \alpha + 1),  (\alpha^2 + 2, 2),(\alpha^2 + 1, \alpha),\\
& (2 \alpha^2, \alpha + 1),  (2 \alpha^2 + \alpha + 2, \alpha), (\alpha^2 + 2 \alpha + 2, \alpha),  (\alpha, \alpha^2 + \alpha + 2),  (2 \alpha + 1, 2 \alpha^2 + 2), (\alpha^2 + \alpha, 2),\\
& (\alpha^2 + \alpha, 2 \alpha^2 + 2),(2 \alpha + 1, 2), (2 \alpha, 2),
 (\alpha^2 + \alpha + 2, \alpha + 2), (\alpha^2 + \alpha + 1, 2 \alpha^2 + 2 \alpha + 1), \\
 &  (2 \alpha^2 + 2 \alpha + 2, \alpha),  (\alpha^2 + 2, 2 \alpha^2 + \alpha + 1), (2 \alpha^2 + 1, \alpha^2 + \alpha + 2), (\alpha^2 + 2 \alpha + 1, 2 \alpha^2 + \alpha + 1), \\
& (2 \alpha, 2 \alpha^2 + \alpha + 1), (2 \alpha + 2, 2 \alpha^2 + 2 \alpha + 1), (\alpha + 1, \alpha^2 + 1),  (2 \alpha^2, \alpha + 2), (2 \alpha^2 + 2, 2 \alpha^2 + \alpha + 1), \\
&(\alpha^2, 2 \alpha^2 + 2), (\alpha^2 + 2 \alpha, 2 \alpha^2 + 2 \alpha + 1),  (\alpha^2 + 2 \alpha + 2, \alpha^2 + 1), (\alpha + 2, \alpha^2 + 2 \alpha + 2), (2 \alpha + 2, 2),\\
& (2 \alpha^2, \alpha^2 + 2 \alpha + 2),  (2 \alpha^2 + \alpha + 2, \alpha^2 + \alpha + 2), (2 \alpha^2 + 2 \alpha + 2, \alpha + 1), (\alpha^2 + 2 \alpha + 2, \alpha + 1), \\
& (2 \alpha^2 + \alpha + 1, 2 \alpha^2 + 2 \alpha + 1), (2 \alpha^2 + \alpha + 2, \alpha + 2), (\alpha^2 + 1, \alpha + 2), (\alpha^2 + 2 \alpha, 2),\\
&  (\alpha^2 + 1, \alpha^2 + \alpha + 2), (2 \alpha^2 + \alpha, \alpha^2 + 1), (2 \alpha^2 + 2 \alpha + 1, 2 \alpha^2 + 2), (2 \alpha^2 + 2 \alpha + 2, \alpha^2 + 1).
 \end{align*}
 }

   Let $f(x)=x^{10}+x^6-x^2$ on $\F_{3^4}$. We list all values of $(c,b)$, for which the $c$-DDT entry  (for some $a$) equals the  $c$-boomerang uniformity of~$7$ (all   $c$-BCT entries of $[1, 2, 3, 4, 5, 6, 7]$ occur):
$
 (\alpha^2 + 2 \alpha + 1, 2 \alpha^3 + 2 \alpha^2 + 1), (\alpha^2 + \alpha + 2, \alpha^3 + \alpha^2), 
 (\alpha^3 + 2 \alpha + 1, \alpha^3 + \alpha^2), (2 \alpha^3 + \alpha^2 + \alpha, 2 \alpha^3 + 2.
 $
 
 Let $f(x)=x^{10}+x^6-x^2$ on $\F_{3^5}$. 
 There are $160$ pairs $(c,b)$, for which the $c$-DDT entry  (for some $a$) equals the  $c$-boomerang uniformity of~$5$ (all   $c$-BCT entries of $[1,2,3,4,5]$ occur). We shall list here only the $140$ different values of $c$:
 \allowdisplaybreaks
{\tiny  \begin{align*}
&  2 \alpha + \alpha^3 + \alpha^4, 2 + 2 \alpha^2 + \alpha^3, 1 + 2 \alpha^3, 1 + \alpha + \alpha^2, 
 2 + \alpha^2 + \alpha^4, 1 + \alpha + 2 \alpha^3, \alpha + \alpha^2, 2 \alpha + \alpha^2, 1 + 2 \alpha, 1 + 2 \alpha + \alpha^2,   \\
 & 1 + 2 \alpha + \alpha^2 + 2 \alpha^3, \alpha^2, 1 + \alpha + \alpha^3, 
 \alpha + 2 \alpha^2 + 2 \alpha^3 + 2 \alpha^4, 2 + 2 \alpha + \alpha^2 + \alpha^4, 
 1 + 2 \alpha + 2 \alpha^2 + 2 \alpha^4, 1 + 2 \alpha + \alpha^3 + \alpha^4,\\
& 2 + \alpha + \alpha^2 + 2 \alpha^3 + 2 \alpha^4, 2 + 2 \alpha, \alpha^2 + \alpha^3,  2 + 2 \alpha + \alpha^2 + 2 \alpha^3 + \alpha^4, \alpha^4, \alpha + 2 \alpha^2 + 2 \alpha^3 + \alpha^4,  1 + \alpha + \alpha^3 + \alpha^4, \alpha^3 + 2 \alpha^4, \\
& 1 + \alpha^3,  2 + \alpha + 2 \alpha^2 + 2 \alpha^3 + 2 \alpha^4, 2 \alpha + 2 \alpha^2 + \alpha^3 + \alpha^4, 
 2 \alpha + 2 \alpha^2 + \alpha^3, 1 + 2 \alpha + 2 \alpha^2 + 2 \alpha^3, 2 \alpha + 2 \alpha^2 + \alpha^4,  2 \alpha + \alpha^2 + \alpha^3, \\
 & 1 + \alpha^2 + 2 \alpha^4, 1 + \alpha^2 + 2 \alpha^3 + 2 \alpha^4, 1 + \alpha + \alpha^2 + \alpha^4, \alpha^2 + \alpha^4, \alpha^2 + \alpha^3 + 2 \alpha^4, 2 + 2 \alpha^4,  2 + 2 \alpha + 2 \alpha^4, 2 \alpha + \alpha^2 + \alpha^4,\\
 &\alpha + \alpha^3 + 2 \alpha^4,  2 + 2 \alpha + 2 \alpha^2 + \alpha^4, 1 + \alpha, 2 + 2 \alpha^2 + \alpha^3 + 2 \alpha^4, 
 \alpha^2 + 2 \alpha^3 + \alpha^4, 2 + \alpha^2 + \alpha^3, \alpha + \alpha^2 + 2 \alpha^4, 
 \alpha^2 + 2 \alpha^3,   \\
 & \alpha + 2 \alpha^4, 1 + \alpha^3 + \alpha^4, 2 + 2 \alpha + 2 \alpha^3 + 2 \alpha^4, 1 + \alpha^2, 2 \alpha + \alpha^2 + 2 \alpha^3 + 2 \alpha^4, 
 2 + 2 \alpha + \alpha^4, 2 + \alpha + \alpha^2 + 2 \alpha^3 + \alpha^4, \\
 &1 + \alpha + \alpha^2 + \alpha^3 + 2 \alpha^4, \alpha + 2 \alpha^2 + \alpha^3 + 2 \alpha^4, 2 + \alpha + \alpha^3, 1 + \alpha + \alpha^3 + 2 \alpha^4, \alpha^3, 
 2 + \alpha + 2 \alpha^2 + 2 \alpha^4, 2 + \alpha^2 + \alpha^3 + 2 \alpha^4,\\
 & 1 + 2 \alpha + \alpha^2 + \alpha^3,  2 \alpha + \alpha^2 + 2 \alpha^3, 2 \alpha + \alpha^4, 1 + \alpha + 2 \alpha^2 + \alpha^3 + 2 \alpha^4, 
 2 + 2 \alpha^2 + 2 \alpha^3 + \alpha^4, 1 + 2 \alpha + \alpha^2 + 2 \alpha^3 + \alpha^4, \\
& 1 + 2 \alpha^3 + 2 \alpha^4, 2 + 2 \alpha^2 + 2 \alpha^3, 1 + 2 \alpha^2 + \alpha^3 + 2 \alpha^4, 
 2 + \alpha^2 + 2 \alpha^4, 2 + 2 \alpha^3, 1 + \alpha + \alpha^2 + 2 \alpha^3 + \alpha^4, 
 \alpha + \alpha^2 + \alpha^3 + \alpha^4, \\
 & 1 + 2 \alpha + 2 \alpha^2 + 2 \alpha^3 + 2 \alpha^4, 
 2 + 2 \alpha^2 + 2 \alpha^3 + 2 \alpha^4, 2 + \alpha^2 + 2 \alpha^3 + 2 \alpha^4, 
 2 + \alpha + 2 \alpha^2 + \alpha^4, \alpha + \alpha^2 + \alpha^3, 2 \alpha + 2 \alpha^4,\\
 &  \alpha + \alpha^2 + 2 \alpha^3, 
 2 \alpha^2, 2 + \alpha + \alpha^2 + 2 \alpha^3, 1 + 2 \alpha + 2 \alpha^2 + 2 \alpha^3 + \alpha^4, 
 1 + \alpha + \alpha^2 + \alpha^3 + \alpha^4, 2 + \alpha^3 + \alpha^4, \alpha + 2 \alpha^2 + \alpha^3, \\
& 2 + \alpha^2 + 2 \alpha^3, \alpha + 2 \alpha^2, \alpha, 1 + 2 \alpha + \alpha^2 + \alpha^3 + \alpha^4, 
 \alpha + 2 \alpha^2 + 2 \alpha^4, 2 + 2 \alpha + 2 \alpha^2 + \alpha^3 + \alpha^4, \alpha + \alpha^3 + \alpha^4,  \\
 & 2 \alpha^2 + \alpha^4, 1 + \alpha + 2 \alpha^2 + 2 \alpha^4, 
 2 \alpha^2 + 2 \alpha^3, 2 \alpha + \alpha^2 + 2 \alpha^3 + \alpha^4, 1 + 2 \alpha + 2 \alpha^4, 
 2 \alpha^3 + 2 \alpha^4, 2 \alpha + \alpha^2 + \alpha^3 + 2 \alpha^4,\\
 & 1 + 2 \alpha^2 + 2 \alpha^3 + 2 \alpha^4, 
 2 + \alpha + 2 \alpha^3 + \alpha^4, \alpha + 2 \alpha^2 + \alpha^4, 2 + \alpha + \alpha^4, 
 2 \alpha^2 + 2 \alpha^3 + \alpha^4, 1 + 2 \alpha + 2 \alpha^3 + 2 \alpha^4, 2 + \alpha + 2 \alpha^3, \\
& 2 + \alpha + 2 \alpha^4, 2 \alpha^2 + \alpha^3 + 2 \alpha^4, \alpha + 2 \alpha^3, 2 + 2 \alpha + 2 \alpha^2 + \alpha^3,
  2 + \alpha + \alpha^2 + 2 \alpha^4, \alpha + \alpha^2 + \alpha^4, 2 + 2 \alpha^3 + \alpha^4, \alpha + \alpha^4, \\
& 1 + 2 \alpha^2 + \alpha^4, \alpha^2 + \alpha^3 + \alpha^4, 2 + 2 \alpha + 2 \alpha^2 + 2 \alpha^3, 
 \alpha + 2 \alpha^3 + \alpha^4, 2 + \alpha + 2 \alpha^2 + 2 \alpha^3 + \alpha^4, 1 + 2 \alpha^3 + \alpha^4, \\
& 2 + 2 \alpha + \alpha^2 + \alpha^3 + 2 \alpha^4, \alpha + \alpha^2 + 2 \alpha^3 + 2 \alpha^4, 
 2 + 2 \alpha + 2 \alpha^2 + \alpha^3 + 2 \alpha^4, 1 + 2 \alpha^4, 2 + \alpha + \alpha^2,  2 + \alpha^2 + 2 \alpha^3 + \alpha^4,\\
& 1 + \alpha + \alpha^2 + \alpha^3, 1 + \alpha + \alpha^2 + 2 \alpha^4, 2 + 2 \alpha^2 + 2 \alpha^4, 
 2 + 2 \alpha + \alpha^2 + \alpha^3, 1 + \alpha^3 + 2 \alpha^4, 2 + 2 \alpha + \alpha^3 + \alpha^4, \\
& \alpha^2 + 2 \alpha^3 + 2 \alpha^4, 2 \alpha^2 + 2 \alpha^4, 2 + 2 \alpha + 2 \alpha^2 + 2 \alpha^3 + 2 \alpha^4, 
 1 + 2 \alpha + \alpha^2 + 2 \alpha^3 + 2 \alpha^4, , 2 \alpha + 2 \alpha^2 + 2 \alpha^3 + \alpha^4.
 \end{align*}
 }
 
 \end{document}